\documentclass[reqno,12pt]{amsart}

\usepackage{amsmath, amsthm, amsfonts, amssymb}

\theoremstyle{plain}
\newtheorem{theorem}{Theorem}
\newtheorem{proposition}[theorem]{Proposition}
\newtheorem{lemma}[theorem]{Lemma}

\newtheorem*{maintheorem}{Theorem}

\theoremstyle{definition}
\newtheorem{example}[theorem]{Example}

\theoremstyle{remark}
\newtheorem{remark}[theorem]{Remark}

\DeclareMathOperator{\res}{Res}

\def\pa{\partial}
\def\de{\delta}

\def\la{\lambda}
\def\bla{{\bar\lambda}}
\def\Z{\mathbb{Z}}	
\def\C{\mathbb{C}}	
\def\R{\mathbb{R}}

\renewcommand{\leq}{\leqslant} 		
\renewcommand{\geq}{\geqslant}

\def\bem#1\enm{\begin{pmatrix}#1\end{pmatrix}} 

\newcommand{\nn}{\nonumber}

\def\beq#1\eeq{\begin{equation}#1\end{equation}}

\def\bea#1\eea{\begin{align}#1\end{align}}   

\def\bean#1\eean{\begin{align*}#1\end{align*}}   

\def\bes#1\ees{\begin{subequations}#1\end{subequations}}

\def\beqa#1\eeqa{\begin{equation} \begin{aligned}#1\end{aligned}\end{equation}}   

\def\cH{\mathcal{H}}
\def\cL{\mathcal{L}}
\def\cU{\mathcal{U}}	
\def\cV{\mathcal{V}}	
\DeclareMathOperator{\ein}{Ein}

\def\S{\mathbb{S}}
\def\l{\la}
\def\bl{\bla}

\def\x{x}
\def\bx{{\bar{x}}}

\newcommand{\Function}[5]{\begin{array}{cccc}
#1:&#2&\longrightarrow&#3\\&#4&\longmapsto&#5\end{array}}

\newcommand{\isomorphism}[4]{\begin{array}{ccc}
#1&\longrightarrow &#2\\
#3&\longmapsto&#4\end{array}}

\def\bspl#1\espl{\begin{equation}\begin{split} #1\end{split}\end{equation}}

\def\cS{\mathcal{S}}

\begin{document}

\title[Infinite-dimensional principal hierarchies]{Principal hierarchies of infinite-dimensional Frobenius manifolds: the extended 2D Toda lattice}
\date{13 January 2015}
\author[G.~Carlet]{Guido Carlet}
\address{Korteweg-de Vries Institute for Mathematics, University of Amsterdam, P.O. Box 94248, 
1090 GE Amsterdam, 
The Netherlands}
\email{g.carlet@uva.nl}
\author[L.~Ph.~Mertens]{Luca Philippe Mertens}
\address{Instituto Nacional de Matematica Pura e Aplicada,
Estrada Dona Castorina 110, 22460-320 Rio de Janeiro, Brasil}
\email{mertens@impa.br}
\begin{abstract}
We define a dispersionless tau-symmetric bihamiltonian integrable hierarchy on the space of pairs of functions analytic inside/outside the unit circle with simple poles at $0$/$\infty$ respectively, which extends the dispersionless 2D Toda hierarchy of Takasaki and Takebe. 
Then we construct the deformed flat connection of the infinite-dimen\-sional Frobenius manifold $M_0$ introduced by Carlet, Dubrovin and Mertens in Math. Ann. 349 (2011) 75--115 and, by explicitly solving the deformed flatness equations, we prove that the extended 2D Toda hierarchy coincides with principal hierarchy of $M_0$.
\end{abstract}

\maketitle

{\small
\tableofcontents
}

\raggedbottom

\section*{Introduction}

The theory of Frobenius manifolds~\cite{D96}, originating as a geometric formulation of the associativity equations of two-dimensional topological field theory~\cite{Wit91, DVV91}, has proven to be an important tool in the study and classification of bi-Hamiltonian tau-symmetric integrable hierarchies of PDEs with one spatial variable~\cite{DZ01}.

The possibility of extending the Frobenius manifolds techniques to the realm of integrable PDEs with two spatial variables has been recently proposed in~\cite{CDM10} where, in collaboration with B.~Dubrovin, we have constructed an infinite-dimensional Frobenius manifold naturally associated with the bi-Hamiltonian structure of the dispersionless 2D Toda hierarchy.
In this work we further develop the program outlined in~\cite{CDM10}. 

First we show that the insights coming from the theory of Frobenius manifolds allow us to solve the problem of finding an extension of the dispersionless 2D Toda hierarchy. This problem has been open since the introduction of the extended~\cite{CDZ04} and extended bigraded~\cite{C07} Toda hierarchies, which are characterized by extra ``logarithmic'' flows.  Such flows are essential in the framework of~\cite{DZ01} and for the description of Gromov-Witten potentials~\cite{DZ04, MT08}, but cannot be obtained by reduction of the 2D Toda flows of~\cite{UT84, TT95}.
To achieve this extension we assume certain analytical properties of the Lax symbols $\la(z)$, $\bla(z)$. First we require that, in contrast with the usual approach of e.g.~\cite{TT95}, they are not just formal power series but genuine holomorphic functions on the exterior/interior part of the unit circle with simple poles at $0/\infty$ respectively. This requirement stems from the necessity of giving a precise meaning to products like $\la^p(z) \bla^q(z)$ which are ill-defined in the case of formal power series. 
We need to consider however not only polynomial expressions of $\la(z)$ and $\bla(z)$ but also logarithms. 
While in the case of extended Toda hierarchies~\cite{CDZ04, C07} the introduction of logarithmic flows does not present serious problems (as one can  define logarithms of formal power series as in~\cite{EY94}), in the case of the 2D Toda hierarchy this straightforward approach is not satisfactory. The correct insights come by considering the relationship of this hierarchy with the Frobenius manifold $M_0$ constructed in~\cite{CDM10}. It turns out that we have to consider logarithms not only of $\la(z)$, $\bla(z)$ but also of $w(z) = \la(z) + \bla(z)$. To make sense of these we impose further conditions on the winding numbers of the analytic curves obtained by restricting these functions to the unit circle in the complex plane. 
These conditions define an open set $M_1$ in the space of pairs of holomorphic Lax symbols. The {\it extended dispersionless 2D Toda hierarchy} is then defined as a family of commuting bi-Hamiltonian vector fields on the loop space $\cL M_1$ which admit a Lax formulation and include the usual flows of the dispersionless 2D Toda hierarchy. 

The main motivation for the definition of this hierarchy comes from the construction~\cite{CDM10} of the infinite-dimensional semisimple Frobenius manifold $M_0$ associated with the standard Poisson pencil of the 2D Toda hierarchy\footnote{Recently this construction has been generalised to other Poisson pencils~\cite{WZ14}, obtaining a family of Frobenius manifolds parametrised by two positive integers $(n,m)$.}.
The Frobenius manifold $M_0$ is defined on the space of pairs of holomorphic Lax symbols $\lambda(z)$, $\bla(z)$ with certain additional conditions (ensuring the invertibility of the metric $\eta$ and the well-posedness of the Riemann-Hilbert problem defining the flat coordinates). We will further assume that $\la(z)$, $\bla(z)$ satisfy the winding numbers condition mentioned above, i.e. we regard $M_0$ as an open subset of $M_1$. 

It is well-known that Frobenius manifolds are naturally associated with a certain class of bi-Hamiltonian dispersionless integrable hierarchies. This is first of all apparent from the fact that any Frobenius manifold $M$ is endowed with a pencil of flat metrics and, consequently, with two compatible hydrodynamic type Poisson structures $\{,\}_1$, $\{,\}_2$ on its loop space~\cite{DZ01}. Moreover a set $A(M)$ of Hamiltonian densities (which define Hamiltonians on $\cL M$, in involution with respect to the first Poisson structure $\{,\}_1$) is naturally associated with the Frobenius manifold $M$. One can single out a distinguished basis of $A(M)$, by constructing a basis of deformed flat coordinates, i.e. of horizontal sections of the basic geometric object associated to a $M$, the deformed flat connection  $\tilde\nabla$. 
Such distinguished basis is fixed by solving in a normal form (Levelt form) the isomonodromic matrix-valued linear system on the deformation parameter $\zeta$ in the $\C$-plane. The corresponding Hamiltonian flows define the so-called Principal hierarchy.
Here we will not review in detail these Frobenius manifold constructions  but rather refer the reader to~\cite{D99, DZ01}.

In the second part of this article we construct the deformed flat connection $\tilde\nabla$ of the Frobenius manifold $M_0$ and we explicitly solve the associated deformed flatness equations. We show that the deformed flat coordinates are indeed a Levelt system, thus we obtain the Principal hierarchy of $M_0$, which coincides with the extended dispersionless 2D Toda hierarchy, when restricted to $\cL M_0$.

Since the Frobenius manifold $M_0$ has a resonant spectrum, the normal form of the fundamental solution of deformed flatness equation in the $\zeta$-plane is not univocally determined but admits an arbitrariness. We exploit such arbitrariness to provide a set of deformed flat coordinates, and a related set of Hamiltonian densities defining the principal hierarchy, which satisfy the orthogonality~\eqref{orthogonnn}. In this way we construct a principal hierarchy with all the good properties expected from the finite-dimensional case.

Let us now state more precisely our results.

The usual Lax formulation of the dispersionless 2D Toda hierarchy is given in terms of two formal Laurent series, the Lax symbols
\[
\la(z) =z + u_0 + \dots, \quad
\bla(z) = \bar u_{-1} z^{-1} + \bar u_0 + \dots ,
\]
and Lax equations~\eqref{eq:TodaLax} which define an infinite set of commuting vector fields on the loop space of formal Lax symbols. 
In order to define a larger set of flows on this loop space we need to impose some analyticity conditions. 
First we require that, instead of formal power series, $\la(z)$, $\bla(z)$ are ``holomorphic Lax symbols'', i.e. $\la(z)$, resp. $\bla(z)$, is a holomorphic function on a neighborhood of closure of the exterior, resp. interior, part of the unit circle in Riemann sphere $\C \cup \{ \infty \}$ admitting a simple pole at $\infty$, resp. $0$, with the normalization condition
$\la(z) = z + O(1) \text{ for } |z| \to \infty$.
Second we impose a ``winding numbers condition'', namely we require that the functions $\la(z)$, $\bla(z)$ and $w(z):=\la(z)+ \bla(z)$, when restricted to the unit circle $|z|=1$, define analytic curves in $\C^\times := \C \backslash \{0\}$ which have winding number around the origin respectively equal to $1$, $-1$ and $1$. We denote by $M_1$ the space of pairs of holomorphic Lax symbols satisfying the winding numbers condition. 

The dispersionless 2D Toda equations define evolutionary flows on pairs of holomorphic Lax symbols depending smoothly on the variable $x$, hence, in particular, on the loop space of $M_1$,
\[
\cL M_1 = C^{\infty} (S^1 ,  M_1) .
\]
We define the {\it extended dispersionless 2D Toda hierarchy} as a set of commuting flows on $\cL M_1$ which include the usual 2D Toda flows. We summarize the first part of our results in the following theorem.

\begin{maintheorem}
Let $Q_{\hat\alpha,p}$, for $\hat\alpha \in \hat\Z = \Z\cup\{u,v\}$ and $p \geq0$, be functions of $\la$, $\bla$ defined by
\bes
\bean
Q_{\alpha, p} &= - \frac{(\la+\bla)^{\alpha+1}(\bla - \la)^p}{(\alpha+1)(2p)!!}   \qquad \text{for} \qquad \alpha \not=-1 ,\\
Q_{-1, p} &= - \frac{(-\la)^p}{p!} \Big( \log \left( 1+\frac{\bla}{\la} \right) + c_p -1 \Big) - \frac{(\bla-\la)^p}{(2p)!!} , \\
Q_{v,p} &= - \frac{(-\la)^p}{p!} \Big(  \log \left( 1+\frac{\bla}{\la} \right)+ c_p -1 \Big)  +
 \frac{\bla^p}{p!} \Big( \log \left( \bla (\la+ \bla) \right) - c_p -1 \Big) ,\\
Q_{u,p} &= \frac{\bla^{p+1}}{(p+1)!}  
\label{Au}
\eean
\ees
where $c_p = 1 + \cdots + \frac1p$ are the harmonic numbers (with $c_0=c_{-1}=0$). 

The Lax equations
\[
\frac{\pa \la}{\ \pa t^{\hat\alpha,p}} = \{ -(Q_{\hat\alpha,p})_- , \la \} ,   \qquad
\frac{\pa \bla}{\ \pa t^{\hat\alpha,p}} = \{ (Q_{\hat\alpha,p})_+  , \bla \} ,
\]
define a tau-symmetric hierarchy of commuting flows on the loop space of $M_1$.

These flows admit a bi-Hamiltonian formulation with respect to the dispersionless 2D Toda Poisson brackets $\{, \}_1$ and $\{,\}_2$
\[
\frac{\pa}{\pa t^{\hat\alpha,p}}  \cdot = \{ \cdot , H_{\hat\alpha,p} \}_1
\]
with recursion relations
\bean
&\{ \cdot , H_{\alpha,p} \}_2 = (\alpha+p+2) \{ \cdot , H_{\alpha,p+1} \}_1 ,\\
&\{ \cdot , H_{v,p} \}_2 = (p+1) \{ \cdot , H_{v,p+1} \}_1 +
2 \{ \cdot , H_{u,p} \}_1 ,\\
&\{ \cdot, H_{u,p} \}_2 = (p+2) \{ \cdot , H_{u,p+1} \}_1 
\eean
for $p\geq-1$, and Hamiltonians
\[
H_{\hat\alpha,p} = \oint_{S^1} \frac1{2 \pi i} \oint_{|z|=1} Q_{\hat\alpha,p+1} (\la(z), \bla(z) ) \ \frac{dz}z \ dx .
\]
The dispersionless 2D Toda flows are finite combinations of the flows defined above.
\end{maintheorem}

The second part of our results, as mentioned above, involves the solution of the deformed flatness equations of the Frobenius manifold $M_0$. 
More precisely let us consider the following functions on $M_0 \times \C^\times$, which are holomorphic in the parameter $\zeta$ in a neighborhood of $\zeta=0$
\bes
\label{thetas}
\bea
\theta_{\alpha}(\zeta) &= - \frac1{2 \pi i} \oint_{|z|=1}  \frac{(\l+\bl)^{\alpha+1}}{\alpha+1} e^{\frac{\bl-\l}2  \zeta} \ \frac{dz}z 
\quad \text{for} \quad \alpha\not= -1 ,\\
\theta_{-1}(\zeta) &=  
-\frac1{2 \pi i} \oint_{|z|=1} \left[ 
e^{- \la \zeta} \left( \log \left(1+\frac{\bla}{\la}\right)+ \ein(-\la\zeta) -1 \right) + e^{\frac{\bla-\la}{2} \zeta} \right] \frac{dz}z ,
\\
\theta_v(\zeta) &=  
\frac1{2 \pi i} \oint_{|z|=1} \Big[ 
- e^{- \la \zeta} \left( \log \left(1+\frac{\bla}{\la}\right) + \ein(-\la\zeta) -1 \right) +\\
&\qquad \qquad+ e^{\bla \zeta} \left( \log \bla(\la + \bla)  - \ein(\bla \zeta) -1\right) \Big] \frac{dz}{z} ,
\\
\theta_u(\zeta) &=  
\frac1{2 \pi i} \oint_{|z|=1}\frac{e^{\bla \zeta}-1}{\zeta} \ \frac{dz}{z} 
\eea
\ees
and the functions   
\bes
\label{yth}
\bea
y_\alpha(\zeta)&= \zeta^{\alpha+\frac12}\theta_\alpha(\zeta),\\
y_v(\zeta)&=\zeta^{\frac{1}{2}}\left(\zeta^{-1}\theta_v(\zeta)+2\log(\zeta)\theta_u(\zeta)\right),\\
y_u(\zeta)&=\zeta^{\frac{1}{2}}\theta_u(\zeta),
\eea
\ees  
which are multivalued in $\zeta$  on $\C^\times := \C\backslash\{0\}$. Here $\ein(z)$ denotes the entire exponential integral function, defined in~\eqref{einfn}.

\begin{maintheorem}
The sequence of functions $\left\{ y_{\hat\alpha}(\zeta) \right\}_{\hat\alpha\in\Z}$ on $M_0 \times \C^\times$ forms a Levelt basis of deformed flat coordinates for $M_0$. 
\end{maintheorem}

The Principal hierarchy is given by the set of vector fields on $\cL M_0$ defined by the Poisson structure associated with the flat metric $\eta$ on $M_0$ and the Hamiltonians 
\[
H_{\hat\alpha,p} = \oint_{S^1} \theta_{\hat\alpha,p+1} \ dx,
\]
where the Hamiltonian densities $\theta_{\hat\alpha,p}$ are obtained by expanding at $\zeta=0$ the analytic part of the Levelt basis of deformed flat coordinates 
\[
\theta_{\hat\alpha}(\zeta) = \sum_{p\geq0} \theta_{\hat\alpha,p} \zeta^p .
\]
Since 
\[
\theta_{\hat\alpha,p} = \frac1{2 \pi i} \oint_{|z|=1} Q_{\hat\alpha,p} \frac{dz}z,
\]
the Hamiltonians of the Principal hierarchy are equal to those of the extended 2D Toda hierarchy defined before.
\begin{maintheorem}
The Principal hierarchy of the Frobenius manifold $M_0$ coincides with the extended dispersionless 2D Toda hierarchy restricted on $\cL M_0$. 
\end{maintheorem}

Notice that recently other examples of infinite-dimensional Frobenius manifolds and associated dispersionless hierarchies have appeared in the literature.  
In~\cite{Rai10} Raimondo has constructed a Frobenius manifold structure on a vector subspace of the space of Schwartz functions $\cS(\R)$ on the real line which is associated with the dispersionless Kadomtsev-Petviashvili hierarchy (dKP).
Wu and Xu~\cite{WX11} have defined a family of Frobenius manifolds on the space of pairs of certain even functions meromorphic in the interior/exterior of the unit disk in $\C$, which are related to the dispersionless  two-component BKP hierarchy. 
In both cases the authors define, essentially by bihamiltonian recursion,  dispersionless hierarchies which extend the original $2+1$ systems. Note that the definition of such hierarchies is somehow simpler since it does not require the construction in terms of  the Levelt normal form of the deformed flat connection, as presented here.


The article is organized as follows: in the first Section we define the extended dispersionless 2D Toda hierarchy on the loop space of holomorphic Lax symbols with certain conditions on the winding numbers. We first recall some basic facts on the dispersionless 2D Toda hierarchy and introduce the analytic setting. Next we give the Lax and bi-Hamiltonian formulation of the extended flows and show that they indeed contain the usual dispersionless flows of the 2D Toda hierarchy. 
In the second Section we study the deformed flat connection associated with the infinite-dimensional Frobenius manifold $M_0$ defined in~\cite{CDM10}. We obtain simple expressions for the metric and its Levi-Civita connection in a new set of ``mixed'' coordinates. Necessary background from the theory of the Frobenius manifold $M_0$ is recalled when necessary. We solve explicitly the deformed flatness equations and prove that our solution provides a Levelt system of deformed flat coordinates. Expanding in the deformation parameter we obtain the Hamiltonian densities of the Principal hierarchy of $M_0$ which coincide with those of the extended 2D Toda hierarchy.
Finally an alternative, though slightly more complicated, system of deformed flat coordinates satisfying an orthogonality condition (cf.~\cite[Theorem 3.6.4]{DZ01}), and the corresponding Principal hierarchy  are presented.

\section{The extended dispersionless 2D Toda hierarchy}

In this section we define an extension of the dispersionless 2D Toda hierarchy introduced by Takasaki and Takebe~\cite{TT95} as the small dispersion limit of the 2D Toda hierarchy of Ueno and Takasaki~\cite{UT84}. 
To perform such extension we assume that Lax functions $\la$, $\bla$ are non-vanishing holomorphic functions on neighborhoods of $z=\infty$, $0$ respectively, which contain the unit circle, and that they satisfy certain analytic assumptions. We begin by recalling the standard formulation of the dispersionless 2D Toda hierarchy.

\subsection{The dispersionless 2D Toda hierarchy}

The dispersionless 2D Toda hierarchy is an infinite set of commuting quasi-linear PDEs for two sets of variables $u_k$, $\bar u_l$ depending on a ``space'' variable $x$ and two series of independent ``time'' variables 
$t = (t_k)_{k\geq 0}$, $\bar t = (\bar t_k)_{k\geq 0}$.  Let the Lax symbols
\[
\la(z,x)=z+\sum_{k\leq 0}u_k(x) z^k, \qquad \bla(z,x)=\sum_{l\geq -1}\bar u_l(x) z^l.
\]
be two formal Laurent series in $z$. The dispersionless 2D Toda hierarchy is defined by the Lax equations
\bes
\label{eq:TodaLax}
\bea
&\frac{\pa \l}{\pa t_n} =  \{ (\l^n)_+ , \l \} , \qquad 
\frac{\pa \bla}{\pa t_n} =  \{ (\l^n)_+ , \bla \} , \\
&\frac{\pa \l}{\pa \bar t_n} =  \{ (\bla^n)_- , \l \} , \qquad 
\frac{\pa \bla}{\pa \bar t_n} =  \{ (\bla^n)_- , \bla \} .
\eea
\ees
The bracket of two functions of the variables $z,x$ is defined by
\[
\{ f(z,x), g(z,x) \} = 
z \frac{\pa f}{\pa z} \frac{\pa g}{\pa x} - z \frac{\pa g}{\pa z} \frac{\pa f}{\pa x} ,
\]
while the notations $(\ )_+$, $(\ )_-$ represent projections taken with respect to the variable $z$:
\[
 \left(\sum_k f_k z^k \right)_+ = \sum_{k\geq 0} f_k z^k, \qquad \left(\sum_k f_k z^k \right)_- = \sum_{k < 0} f_k z^k.  
\]

The equations ~\eqref{eq:TodaLax} are formal Laurent series in $z$: each coefficient defines an evolutionary quasi-linear equation involving a finite number of dependent variables $u_k$, $\bar u_l$. Such equations have the remarkable property of defining commutative flows
\[
\left[ \frac{\pa}{\pa s_n} , \frac{\pa}{\pa s_m} \right] = 0 
\]
for $s_n$ equal to either $t_n$ or $\bar t_n$. 

The flows~\eqref{eq:TodaLax} admit a bi-Hamiltonian formulation \cite{C05}
\bes
\label{toda-rec}
\bea
&\frac{\pa}{\pa t_n}\cdot = \{\cdot,H_n\}_1 = -\{\cdot,H_{n-1}\}_2 , \\
&\frac{\pa}{\pa \bar t_n}\cdot= \{\cdot, \bar H_n\}_1 
= \{\cdot ,\bar H_{n-1}\}_2
\eea
\ees
with Hamiltonians given by 
\[
H_n = - \int \res \frac{\la^{n+1}}{n+1} \frac{dz}z dx , \qquad
\bar H_n = - \int \res \frac{\bla^{n+1}}{n+1} \frac{dz}z dx 
,\]
where the residue of a formal series is  
$\res \sum_k f_k z^k \frac{dz}z = f_0$. 
The hydrodynamic type Poisson brackets $\{ , \}_1$ and $\{ , \}_2$ are compatible, i.e. any their linear combination is still a Poisson bracket. These Poisson brackets have been defined in~\cite{C05}. Their definition is recalled below in Proposition~\ref{prop:poi}.

\subsection{Analytic setting}

Denote $D_0$ the closed unit disc in the Riemann sphere $\C \cup \{ \infty\}$, $D_{\infty}$ the closure of the complement of  $D_0$ and $\S^1 = D_0 \cap D_\infty$ the unit circle.   
For a compact subset $K$ of the Riemann sphere, denote by $\cH(K)$ the space of holomorphic functions on $K$, i.e. functions which extend holomorphically to a neighbourhood of $K$. 

For each $p\in\Z$ the space of holomorphic functions on a neighborhood of $\S^1$ splits in a direct sum
\[
\cH(\S^1) = z^p \cH(D_0) \oplus z^{p-1} \cH(D_\infty) 
.\]
The projections 
\[
(\ )_{\geq p} : \cH(\S^1) \rightarrow z^p \cH(D_0) , \qquad
(\ )_{\leq p-1} : \cH(\S^1) \rightarrow z^{p-1} \cH(D_\infty)
\]
are given by
\bean
&(f)_{\geq p} (z)  = \sum_{k\geq p} f_k z^k = \frac{z^p}{2 \pi i} \oint_{|z|<|\zeta|} \frac{\zeta^{-p} f(\zeta)}{\zeta - z} \ d\zeta, \\
&(f)_{\leq p-1} (z)  = \sum_{k\leq p-1} f_k z^k = -\frac{z^p}{2 \pi i} \oint_{|z|>|\zeta|} \frac{\zeta^{-p} f(\zeta)}{\zeta - z} \ d\zeta
\eean
for $f(z) = \sum_{k\in\Z} f_k z^k \in \cH(\S^1)$. As usual $(\ )_+ = (\ )_{\geq0}$ and $(\ )_- =(\ )_{\leq-1}$. The symbol $(\ )_k: \cH(\S^1) \to \C$ denotes the coefficient of $z^k$ in the Laurent series expansion, i.e. $(f)_k = \frac1{2 \pi i} \oint_{|z|=1}  f(z) z^{-k} \frac{dz}z$.

Define the infinite-dimensional manifold $M$ as the affine subspace 
\[
M = 
\big\{ (\la(z), \bla(z)) \in z \cH(D_\infty) \oplus \frac1z \cH(D_0)
\ \big| \ \la(z) = z + O(1) \text{ for } z \to \infty \big\} 
\]
in the direct sum of the vector spaces $z \cH(D_\infty)$ and  
$\frac1z \cH(D_0)$. We will sometimes refer to $M$ as the space of pairs of holomorphic Lax symbols.

For $(\la, \bla) \in M$, the functions $\la(z)$, $\bla(z)$  have the following Laurent series expansions
\[
\la(z)=z+ \sum_{k\leq 0}u_k z^k , \qquad 
\bla(z)=\sum_{k\geq -1}\bar u_k z^k
\] 
at $\infty$ and $0$ respectively.

The tangent space at a point $\hat\la=(\la(z),\bla(z)) \in M$ will be identified with a space of pairs of functions
\[
T_{\hat\lambda}M \cong
\cH(D_{\infty}) \oplus \frac1z \cH(D_0) 
\]
where a vector $\pa_X$ is associated with the pair $X = (X(z),\bar{X}(z))$ given by $X(z)=\pa_X \la(z)$ and $\bar{X}(z)=\pa_X \bla(z)$.

The cotangent bundle at $\hat\la \in M$ will be given also by a space of pairs of functions 
\[
T^*_{\hat\lambda}M \cong
\cH(D_{0}) \oplus z\cH(D_{\infty}) 
\]
by representing a covector $\alpha$ as the pair $(\alpha(z),\bar\alpha(z))$ by the residue pairing
\beq
\label{residue_pairing}
<\alpha, X> = \frac1{2 \pi i} \oint_{|z|=1} \big[
\alpha(z) X(z) + \bar\alpha(z) \bar X(z) 
\big]
\, \frac{d z}{z} .
\eeq
%
Note that we are using a definition of residue pairing which is slightly different from the one in~\cite{CDM10}, in particular here the measure is $dz/z$ instead of $dz$, and as a consequence, in the representation of a covector as a pair of functions $(\alpha(z), \bar\alpha(z))$, we allow $\bar	\alpha(z)$ to have a simple pole rather than $\alpha(z)$. This difference is of course immaterial, but consistent with the usual dispersionless 2D Toda  formulation. 

A point in the loop space $\cL M$ of smooth maps from $S^1$ to the $M$ is given by a pair of functions $(\la(z,x) , \bla(z,x))$. Here $S^1 = \R \mod 2 \pi$, while the symbol $\S^1$ will denote the unit circle in $\C$. 

A tangent vector at a point $(\la(z,x),\bla(z,x)) \in \cL M$ is clearly identified with a map from $S^1$ to $\cH(D_\infty) \oplus \frac1z \cH(D_0)$ and a $1$-form with a map from $S^1$ to $\cH(D_0)\oplus z\cH(D_\infty)$. The pairing of a vector 
$X = (X(z,x), \bar X(z,x))$ and a $1$-form $\alpha=(\alpha(z,x),\bar\alpha(z,x))$ is
\[
< \alpha, X > = \frac1{2 \pi i} \oint_{S^1} \oint_{|z|=1}  \big[ \alpha(z,x) X(z,x) + \bar\alpha(z,x) \bar X(z,x) \big] \, \frac{dz}z dx ,
\]
which is the natural extension of the pairing~\eqref{residue_pairing}.

The equations \eqref{eq:TodaLax} defining the 2D Toda flows specify, for each $n>0$, a vector field over $\cL M$. 
Indeed, note that equations~\eqref{eq:TodaLax} are of the form 
\[
(\pa_t \la, \pa_t \bla )  = (\{ -Q_-, \la \},\{ Q_+, \bla\})
\] 
where $Q=\la^n$ or $\bla^n$. For $(\la, \bla) \in \cL M$, at fixed $x\in S^1$, we have $Q(z)\in\cH(\S^1)$ and we can easily check that this implies that $\{ Q_-, \la \} \in \cH(D_\infty)$ and $\{ Q_+, \bla\} \in \frac1z \cH(D_0)$. Hence $(\{ -Q_-, \la \},\{ Q_+, \bla\}) \in T_{\hat\la} \cL M$.

Recall that a Poisson bracket $\{,\}_i$ of two local  functionals $F$, $G$ on $\cL M$ is written in terms of a Poisson operator $P_i$ from the cotangent to the tangent space of $\cL M$ as follows
\[
\{ F, G \}_i = < dF , P_i(dG) >  .
\]
The bi-Hamiltonian structure of the 2D Toda hierarchy, in the dispersive case, was defined in~\cite{C05} by using R-matrix theory applied to an algebra of pairs of difference operators. We recall here the formulas for the dispersionless limit of the Poisson brackets, which were obtained in~\cite{CDM10}, in the analytic setting. 
\begin{proposition}
\label{prop:poi}
The maps $P_i : T^* \cL M \to T \cL M$ define compatible Poisson brackets on $\cL M$. Such maps, given a $1$-form $\hat\omega=(\omega,\bar\omega) \in T^*_{\hat\la} \cL M$ at $\hat\la=(\la, \bla) \in \cL M$, are defined by 
\bean
P_1(\hat\omega) &= \big( - \{ \la, (\omega-\bar\omega)_- \}+
( \{\la, \omega \} + \{ \bla, \bar\omega \} )_{\leq 0} , \\
&\qquad \{\bla, (\omega-\bar\omega)_+ \}+
( \{\la, \omega \} + \{ \bla, \bar\omega \} )_{> 0} \big) \\
P_2(\hat\omega) &= \big( \{\la, (\la\omega +\bla\bar\omega)_-\} -
\la ( \{\la, \omega \} + \{ \bla, \bar\omega \} )_{\leq 0} + 
z \la' \phi_x , \\
&\qquad - \{\bla, (\la\omega, +\bla\bar\omega)_+\} +
\bla ( \{\la, \omega \} + \{ \bla, \bar\omega \} )_{> 0} +
z \bla' \phi_x \big)
\eean
where
\[
\phi_x = \frac1{2 \pi i} \oint_{|z|=1} \big( \{\la, \omega \} + \{ \bla, \bar\omega \} \big) \frac{dz}z .
\]
\end{proposition}

The Hamiltonians
\[
H_n = - \frac1{2 \pi i} \int_{S^1} \oint_{|z|=1}  \frac{\la^{n+1}}{n+1} \frac{dz}z dx , \qquad
\bar H_n = -\frac1{2 \pi i} \int_{S^1} \oint_{|z|=1}   \frac{\bla^{n+1}}{n+1} \frac{dz}z dx 
\]
define local functionals on $\cL M$ and generate the Hamiltonian vector fields~\eqref{eq:TodaLax} according to~\eqref{toda-rec}. Summarizing well known facts in this analytic setting:
\begin{proposition}
The dispersionless 2D Toda hierarchy equations~\eqref{eq:TodaLax} define a set of bi-Hamiltonian commuting vector fields on $\cL M$, with respect to the Poisson brackets $\{,\}_i$ and with recursion relations~\eqref{toda-rec}. 
\end{proposition}

Note that we have chosen to represent $1$-forms on $M$ by elements of $\cH(D_0) \oplus z\cH(D_\infty)$ using the pairing~\eqref{residue_pairing}. One can more generally represent a $1$-form by a pair of functions in $\cH(\S^1)\oplus \cH(\S^1)$: the $1$-form does not change by adding to the representative an element in $\frac1z \cH(D_\infty) \oplus z^2 \cH(D_0)$ (recall that $z^p\cH(D_\infty)$ and $z^p \cH(D_0)$ are seen here as subspaces of $\cH(\S^1)$, for any $p\in\Z$). The freedom of choosing the representative extends of course to the $1$-forms on the loop space. Later we will need the following easy to check observation:
\begin{lemma}
\label{repres}
For $i=1,2$, the Poisson map $P_i$ maps a $1$-form $\hat\omega$ to a vector $X = P_i(\hat\omega)$ which is independent of the choice of the representative for $\hat\omega \in T_{\hat\la} \cL M$ in $\cH(\S^1)\oplus \cH(\S^1)$.
\end{lemma}

\subsection{The extended hierarchy: Lax formulation}
We now impose extra analyticity conditions on Lax functions which allow us to extend the 2D Toda hierarchy with new flows involving products of $\la$, $\bla$ and their logarithms. 

Let us define the manifold $M_1$ as the open subset of $M$ given by pairs of functions $(\la(z), \bla(z)) \in M$ which satisfy the following winding numbers condition:
\begin{quote}
the functions $w(z):= \la(z) + \bla(z)$, $\la(z)$ and $\bla(z)$ when restricted to the unit circle $\S^1 := \{ z \in \C \text{ s.t. }  |z|=1\}$ define analytic curves in $\C^\times$ with winding number  around $0$ respectively equal to $1$, $1$ and $-1$. 
\end{quote}

Since $M_1$ is an open subset of $M$ we can represent $T M_1$ and $T^* M_1$ as before.

The flows of the extended 2D Toda hierarchy are defined by the following Lax representation
\beq
\label{eq:lax}
\frac{\pa \la}{\ \pa t^{\hat\alpha,p}} = \{ -(Q_{\hat\alpha,p})_- , \la \} ,   \qquad
\frac{\pa \bla}{\ \pa t^{\hat\alpha,p}} = \{ (Q_{\hat\alpha,p})_+  , \bla \} ,
\eeq
for $\hat\alpha \in \hat\Z = \Z\cup\{u,v\}$ and $p\geq 0$. The $Q_{\hat\alpha,p}$ are functions of $\la$, $\bla$ defined by the formulas
\bes
\label{Qfns}
\bea
Q_{\alpha, p} &= - \frac{(\la+\bla)^{\alpha+1} (\bla - \la)^p}{(\alpha+1) (2p)!!}   \qquad \text{for} \qquad \alpha \not=-1 ,\\
Q_{-1, p} &= - \frac{(-\la)^p}{p!} \Big( \log \left( 1+\frac{\bla}{\la} \right) + c_p -1 \Big) - \frac{(\bla-\la)^p}{(2p)!!} , \\
Q_{v,p} &= - \frac{(-\la)^p}{p!} \Big(  \log \left( 1+\frac{\bla}{\la} \right)+ c_p -1 \Big)  +\nn\\
& \quad + \frac{\bla^p}{p!} \Big( \log \left( \bla (\la+ \bla) \right) - c_p -1 \Big) ,\\
Q_{u,p} &= \frac{\bla^{p+1}}{(p+1)!}  
\label{Au}
\eea
\ees
where $c_p = 1 + \cdots + \frac1p$ are the harmonic numbers (with $c_0=c_{-1}=0$). 
%

Let us consider the well-posedness of the Lax equations. We have seen above that each 2D Toda evolutionary flow defines a vector field of the form 
\[
(\{ -Q_-, \la \} , \{ Q_+ , \bla \} ) \in T_{\hat \la} \cL M
\]
over $\cL M$, for $Q= \la^n$ or $\bla^n$. This is based on the fact that $Q$ is an entire function of $\la$ or $\bla$, hence, by composition with $\la(z)$ or $\bla(z)$ it gives a function $Q(z)$ holomorphic in a neighborhood of the unit circle. This in turn implies that the projections make sense and that the dispersionless Lax equations define a vector field. 

For more general functions $Q(\la, \bla)$, which e.g. might not be holomorphic on the whole $\C^2$, the projections appearing in the Lax equations do not make sense for every $\hat\la \in \cL M$. Hence we need to impose extra conditions on the functions $\la$, $\bla$, as we can see from the following simple general observation.

\begin{lemma}
Let $Q(\la, \bla)$ be a multivalued holomorphic function on an open subset of $\C^2$ and $M'$ an open subset of $M$ defined by imposing extra conditions on $(\la,\bla)\in M$, such that 
\beq
\label{basic-cond}
Q(z) = Q(\la(z), \bla(z)) \in \cH(\S^1) \text{ for any } 
(\la,\bla)\in M'
.
\eeq
Then the Lax equations 
\[
\frac{\pa \la}{\ \pa t} = \{ -Q_- , \la \} ,   \qquad
\frac{\pa \bla}{\ \pa t} = \{ Q_+  , \bla \} ,
\]
give a well-defined vector field on $\cL M'$.
\end{lemma}

The conditions defining the manifold $M_1$ clearly imply that the property~\eqref{basic-cond} is satisfied for $Q$'s of the form~\eqref{Qfns}. For example the logarithmic part of $Q_{-1,p}$ can be written as 
\[
\log\left( 1+ \frac\bla\la \right) = 
\log \frac{\la+\bla}{z} - \log\frac\la{z} .
\]
The winding number condition on $\la + \bla$ is equivalent to the fact that $\frac{\la+\bla}z$ can be lifted to a map from $\S^1$ to the Riemann surface of the logarithm, i.e. the universal covering of $\C^\times$. Therefore $\log \frac{\la(z)+\bla(z)}{z}$, as a function of $z$, is in $\cH(\S^1)$. A similar reasoning shows that $\log\frac{\la(z)}{z} \in \cH(\S^1)$, hence $Q_{-1,p}(\la(z) , \bla(z) ) \in \cH(\S^1)$ for $(\la(z), \bla(z)) \in M_1$. In conclusion the Lax equations~\eqref{eq:lax} provide well-defined vector fields on $\cL M_1$.

Let us observe that the winding number conditions on $\la(z)$ and $\bla(z)$ are equivalent to the fact that $\la(z)$, $\bla(z)$ are not vanishing on $D_\infty$, $D_0$ respectively and that the leading term of $\bla(z)$ is non-zero.
\begin{lemma}
\label{windingzeros}
Let $(\la(z), \bla(z)) \in M$. The function $\la(z)$ restricted to $\S^1$ parametrizes an analytic curve in $\C^\times$ of winding number $1$ if and only if $\la(z)$ is non-vanishing in $D_\infty$. The function $\bla(z)$ restricted to $\S^1$ parametrizes an analytic curve in $\C^\times$ of winding number $-1$ if and only if $\bla(z)$ is non-vanishing in $D_0$ and the leading coefficient $\bar u_{-1}$ in  $\bla(z) = \bar u_{-1} z^{-1} + O(1)$ for $|z|\to0$ is non-zero. 
\end{lemma}
\begin{proof}
The winding number of the curve parametrized by $\bla:\S^1\to\C^\times$ is given by the logarithmic residue 
\[
\frac1{2 \pi i} \oint_{|z|=1} \frac{\bla'(z)}{\bla(z)} dz =
N + N_0
\]
where $N\geq0$ is the number of zeros (counted with their multiplicity) of $\bla(z)$ in $D_0\backslash\{0\}$ and $N_0\geq-1$ is the order of $\bla(z)$ at $z=0$. Clearly if the winding number is $-1$ we must have $N=0$ and $N_0=-1$. The converse is obvious as is the statement on $\la(z)$.
\end{proof}

For $(\la(z), \bla(z))\in M_1$ this observation implies that $\frac{\la(z)}z$ (resp. $z\bla(z))$ maps $D_\infty$ (resp. $D_0$) to a bounded subset of $\C^\times$ hence
\bea
\label{logs}
\log\frac{\la(z)}z\in\frac1z\cH(D_\infty), \quad
\log z\bla(z) \in z \cH(D_0).
\eea

Commutativity as usual follows from the Zakharov-Shabat equations, which turn out to be quite simple.
\begin{proposition}
The Lax equations~\eqref{eq:lax} imply that the following ZS zero curvature equations hold
\beq
\label{zs}
\frac{\pa Q_{\hat\alpha,p}}{\pa t^{\hat\beta,q}} - \frac{\pa Q_{\hat\beta,q}}{\pa t^{\hat\alpha,p}} + \{ (Q_{\hat\alpha,p})_+ , (Q_{\hat\beta,q})_+  \} -
\{ (Q_{\hat\alpha,p})_- , (Q_{\hat\beta,q})_-  \} =0 .
\eeq
\end{proposition}
\begin{proof}
To avoid cumbersome notations let 
$Q := Q_{\hat\alpha,p}$ and $\tilde Q := Q_{\hat\beta, q}$. Computing
\[
\frac{\pa Q}{\pa t^{\hat\beta,q}}= Q_\la \{ - \tilde Q_- , \la \} 
+Q_\bla \{ \tilde Q_+, \bla \},
\]
substituting in the last term $\tilde Q_+ = \tilde Q - \tilde Q_-$ and using the Leibniz rule we get
\[
\frac{\pa Q}{\pa t^{\hat\beta,q}}=\{ - \tilde Q_-, Q \} + Q_\bla \{ \tilde Q, \bla \}
\]
hence
\[
\frac{\pa Q}{\pa t^{\hat\beta,q}} - \frac{\pa \tilde Q}{\pa t^{\hat\alpha,p}} = \{ - \tilde Q_-, Q \} + Q_\bla \{ \tilde Q, \bla \}
-\{ - Q_-, \tilde Q \} - \tilde Q_\bla \{ Q, \bla \} .
\]
Rewriting the last term in the right-hand side as 
\[
- \tilde Q_\bla Q_\la \{ \la, \bla \} = - Q_\la \{ \la , \tilde Q \}
\]
and using the Leibniz rule again we get
\[
\frac{\pa Q}{\pa t^{\hat\beta,q}} - \frac{\pa \tilde Q}{\pa t^{\hat\alpha,p}} = \{ - \tilde Q_-, Q \} -\{ - Q_-, \tilde Q \} 
+\{ \tilde Q ,Q \} .
\]
A simple rearrangement of the right-hand side gives the desired result.
\end{proof}

The commutativity of the flows~\eqref{eq:lax} follows from~\eqref{zs}. For example 
\[
\left[ \pa_{t^{\hat\alpha,p}} , \pa_{t^{\hat\beta,q}} \right] \la(z) = 
\left\{ - \frac{\pa (Q_{\hat\alpha,p})_-}{\pa t^{\hat\beta,q}} + \frac{\pa (Q_{\hat\beta,q})_-}{\pa t^{\hat\alpha,p}}  +
\{ (Q_{\hat\alpha,p})_- , (Q_{\hat\beta,q})_-  \}
, \la \right\} =0
\]
because the first term in the big curly bracket is given by the projection to $\frac1z \cH(D_\infty)$ of the left-hand side of~\eqref{zs}.

Let us consider a couple of explicit examples of the extended flows.
\begin{example}
From~\eqref{Qfns} we have that 
\[
Q_{v,0} = \log \la + \log \bla  = \log \frac\la{z} + \log z \bla.
\]
Using~\eqref{logs}, we can easily compute the Lax equations
\[
\frac{\pa\la}{\pa t^{v,0}} 
= \{ - (Q_{v,0})_- , \la \} 
= \{ - \log \frac\la{z} , \la \} 
= \{ \log z , \la \} = \frac{\pa \la}{\pa x} 
\]
and similarly
\[
\frac{\pa\bla}{\pa t^{v,0}} = \frac{\pa \bla}{\pa x}
.\]
Therefore the time $t^{v,0}$ corresponds to the space variable $x$, as expected from the Frobenius manifold construction in the next section. Note that the $x$-translation was not present among the original dispersionless Toda Lax flows.  
\end{example}

\begin{example} 
Consider now the Lax equations for the nontrivial time $t^{v,1}$. From~\eqref{Qfns} we get
\[
Q_{v,1} = \la \log \left( 1 + \frac\bla\la \right)  
+ \bla \log ( \la \bla + \bla^2 ) - 2 \bla .
\] 
For $(\la, \bla)\in M_1$ we can evaluate the projections of 
$Q_{v,1}(z) = Q_{v,1}(\la(z), \bla(z))$. 
To this aim is convenient to rewrite the previous expression as
\[
Q_{v,1} = (\la+\bla) \log \frac{\la+\bla}{z} - \la \log \frac\la{z} 
+ \bla ( \log z \bla -2 )
.\]
Using the identities
\[
\left( \la \log \frac\la{z} \right)_+ = u_0, \quad
\left( \bla \log z \bla \right)_- = ( \bar u_{-1} \log \bar u_{-1} )z^{-1} 
\]
we easily obtain 
\[
(Q_{v,1})_- = \left( (\la+\bla) \log \frac{\la+\bla}{z} \right)_- 
-\la \log\frac\la{z} + u_0 + \bar u_{-1} ( \log \bar u_{-1} -2) z^{-1} 
\]
and then 
%
\bean
\frac{\pa\la}{\pa t^{v,1}} 
&= \{ -(Q_{v,1})_- , \la \} \\
&= \left(
- \left( z ( \la_z + \bla_z ) \log\frac{\la+\bla}z \right)_- 
- u_0 + \bar u_{-1} z^{-1} \log \bar u_{-1} 
\right) \la_x \\
&+ \left( 
\left( z ( \la_x + \bla_x ) \log\frac{\la+\bla}z \right)_- 
+ \frac{\pa\bar u_{-1}}{\pa x} \log \bar u_{-1} 
\right) \la_z
.\eean
With a similar computation one can obtain $\frac{\pa\bla}{\pa t^{v,1}}$.
\end{example}

\subsection{The extended hierarchy: bi-Hamiltonian formulation}

We now show that the evolutionary flows of the extended 2D Toda hierarchy are bi-Hamiltonian with respect to the Poisson structures $P_i$.

The Hamiltonians are functionals on $\cL M_1$ given by
\[
H_{\hat\alpha,p} = \oint_{S^1} h_{\hat\alpha,p} \ dx , 
\quad
\hat\alpha \in \hat\Z, \quad p \geq -1
\]
where the Hamiltonian densities $h_{\hat\alpha,p}$ are expressed in terms of the functions $Q_{\hat\alpha,p}$ defined in~\eqref{Qfns} as
\[
h_{\hat\alpha,p} = 
\frac1{2 \pi i} \oint_{|z|=1} Q_{\hat\alpha,p+1}(\la(z), \bla(z)) \, \frac{dz}z .
\]

\begin{proposition}
The flows of the extended 2D Toda hierarchy are bi-Hamiltonian w.r.t. the Poisson brackets $\{,\}_i$ 
\[
\frac{\pa}{\pa t^{\hat\alpha,p}}  \cdot = \{ \cdot , H_{\hat\alpha,p} \}_1
\]
with Hamiltonians $H_{\hat\alpha,p}$ and the recursion relations
\bean
&\{ \cdot , H_{\alpha,p} \}_2 = (\alpha+p+2) \{ \cdot , H_{\alpha,p+1} \}_1 ,\\
&\{ \cdot , H_{v,p} \}_2 = (p+1) \{ \cdot , H_{v,p+1} \}_1 +
2 \{ \cdot , H_{u,p} \}_1 ,\\
&\{ \cdot, H_{u,p} \}_2 = (p+2) \{ \cdot , H_{u,p+1} \}_1 
\eean
for $p\geq-1$. Moreover $H_{\hat\alpha,-1}$ are Casimirs of $\{,\}_1$.
\end{proposition}

\begin{proof}
The Hamiltonians $H_{\hat\alpha,p}$ are of the form
\[
H = \frac1{2\pi i} \oint_{S^1}   \oint_{|z|=1} Q(\la, \bla) \, \frac{dz}z \ dx ,
\]
where the function $Q(z)=Q(\la(z),\bla(z))$ is in $\cH(\S^1)$, when $(\la,\bla) \in M$. 

From Lemma~\ref{repres} it follows that the differential of $H$ 
\[
dH = \left( \Big( \frac{\pa Q}{\pa\la} \Big)_{\geq0} , \Big( \frac{\pa Q}{\pa\bla} \Big)_{\leq1} \right)
\in \cH(D_0) \oplus z\cH(D_\infty)
\]
can be equivalenty represented by
\[
dH = \left( \frac{\pa Q}{\pa\la} , \frac{\pa Q}{\pa\bla} \right)
\in \cH(\S^1) \oplus \cH(\S^1) .
\]
For any function $Q(\la, \bla)$ we have that 
\[
\{ \la, Q_\la \} + \{ \bla , Q_\bla \} = 0 
\]
hence, substituting in the Poisson maps given in Proposition \ref{prop:poi}, we obtain
\bean
P_1(dH) &= \big( \{\la , ( Q_\bla - Q_\la )_- \}, 
\{ \bla, -( Q_\bla - Q_\la )_+ \} \big), \\
P_2(dH) &= \big( \{\la , (\la Q_\la + \bla Q_\bla)_-\},
\big( -\{\bla , (\la Q_\la + \bla Q_\bla)_+\} \big) ,
\eean
where we have denoted $Q_\la = \frac{\pa Q}{\pa \la}$ and $Q_\bla = \frac{\pa Q}{\pa \bla}$.

One can check directly, using the homogeneity properties of $Q_{\hat\alpha,p}$, that
\[
\left(\frac{\pa}{\pa\bla} - \frac{\pa}{\pa\la} \right)  Q_{\hat\alpha,p} = Q_{\hat\alpha,p-1}
\]
and
\bean
\left( \la \frac{\pa}{\pa\la} + \bla \frac{\pa}{\pa\bla} \right) Q_{\alpha,p} &= (\alpha + p + 1) Q_{\alpha,p} ,\\
\left( \la \frac{\pa}{\pa\la} + \bla \frac{\pa}{\pa\bla} \right) Q_{v,p} &= p Q_{v,p} +2 Q_{u,p-1} ,\\
\left( \la \frac{\pa}{\pa\la} + \bla \frac{\pa}{\pa\bla} \right) Q_{u,p} &= (p+1) Q_{u,p} ,
\eean
for $\hat\alpha\in\hat\Z$, $\alpha\in\Z$ and $p\geq0$. In these formulas we have assumed
\[
Q_{-1,-1} = -\frac1{\la(z)} , \quad Q_{v,-1} = \frac1{\bla(z)} - \frac1{\la(z)}, \quad Q_{u,-1}= 1, 
\]
and $Q_{\alpha,-1}=0$ for $\alpha\not=-1$.

The Lax equations \eqref{eq:lax} and the recursion relations are an easy consequence of these formulas.

The fact that $H_{v,-1}$ is a Casimir follows from
\begin{align*}
P_1(d H_{v,-1}) &= \left( \{ \la, (Q_{v,-1})_-\}, \{ \bla, -( Q_{v,-1})_+ \}\right) \\
&=\left( \{ \la, -\frac1{\la(z)} \}, \{ \bla, - \frac1{\bla(z)} \}\right) = 0 
\end{align*}
and a similar computation holds for $H_{\-1,-1}$. Note that in the last formula it is essential that $\frac1{\la(z)} \in \frac1z\cH(D_\infty)$ and $\frac1{\bla(z)} \in z\cH(D_0)$ i.e. that $\la(z)$ and $\bla(z)$ do not have zeros.
\end{proof}

\begin{remark}
Note that for $\alpha\geq0$ the Lenard-Magri bi-Hamiltonian recursion starts from the Casimir $H_{\alpha,-1}$ of $\{,\}_1$ and induces the infinite chains of Hamiltonians $H_{\alpha, p}$, $p\geq0$. On the other hand for $\alpha\leq-1$ the Lenard-Magri chain starting from the Casimir $H_{\alpha,-1}$ stops after $-\alpha-1$ steps at the Casimir $H_{\alpha,-\alpha-2}$ of the second Poisson bracket. The remaining Hamiltonians $H_{\alpha, p}$ for $p\geq-\alpha-1$, $\alpha\geq-1$ are included in chains starting from $H_{\alpha, -\alpha-1}$ which are not Casimirs. Similarly the chain $H_{u,p}$ starts from the Casimir $H_{u,-1}$ of the first Poisson bracket and the chain $H_{v,p}$ starts from $H_{v,0}$ which is not a Casimir. Finally note that the Hamiltonians $H_{-1,-1}$ and $H_{v,-1}$ are common Casimirs of $\{,\}_1$ and $\{,\}_2$. 
\end{remark}

\begin{remark}
Let us finally observe that the extended 2D Toda hierarchy defined above really extends the usual dispersionless 2D Toda. Indeed,
the Hamiltonians $H_n$, $\bar H_n$ of the dispersionless 2D Toda hierarchy are finite combinations of the Hamiltonians of the dispersionless extended 2D Toda hierarchy: 
\bean
H_n &= (-1)^{n} \, n! \, H_{u,n-1} + \sum_{l=0}^n \frac{n! \ ((-1)^l - (-1)^{n+1})}{(n-l)!\ 2^{n-l+1}} H_{n-l,l-1}, \\
\bar H_n &= - n! \ H_{u,n-1} .
\eean

\end{remark}


%
\subsection{Tau symmetry}
The hamiltonian densities $h_{\hat\alpha, p}$ of the extended dispersionless 2D Toda are tau symmetric, as can be easily proved using the Zakharov-Shabat equations.

\begin{proposition}
The Hamiltonian densities $h_{\hat\alpha,p}$ satisfy the tau-symmetry
\[
\frac{\pa h_{\hat\alpha,p-1}}{\pa t^{\hat\beta,q}} = 
\frac{\pa h_{\hat\beta,q-1}}{\pa t^{\hat\alpha,p}} 
\]
for any $\hat\alpha, \hat\beta \in \hat\Z$ and $p,q\geq0$.
\end{proposition}
\begin{proof}
By equation~\eqref{zs} we have that
\[
\frac{\pa h_{\hat\alpha,p-1}}{\pa t^{\hat\beta,q}} -
\frac{\pa h_{\hat\beta,q-1}}{\pa t^{\hat\alpha,p}} =
\frac1{2 \pi i} \oint_{|z|=1} \left[ 
\frac{\pa Q_{\hat\alpha,p}}{\pa t^{\hat\beta,q}} - \frac{\pa Q_{\hat\beta,q}}{\pa t^{\hat\alpha,p}} \right] \frac{dz}z
\]
is equal to
\[
\frac1{2 \pi i} \oint_{|z|=1} \left[ 
- \{ (Q_{\hat\alpha,p})_+ , (Q_{\hat\beta,q})_+  \} +
\{ (Q_{\hat\alpha,p})_- , (Q_{\hat\beta,q})_-  \} \right] \frac{dz}z 
\]
which clearly vanishes. 
\end{proof}

\begin{remark}
Note that in the proof of the tau-symmetry we have not used the explicit form of the Hamiltonians: the only relevant property is that the Lax equations for the time $t^{\hat\alpha,p}$ and the Hamiltonian density $h_{\hat\alpha,p-1}$ are written in terms of the {\it same} function $Q_{\hat\alpha,p}$. This requires the choice of a proper normalization of the flows.
\end{remark}
%
%
%

\section{The principal hierarchy of $M_0$}

In~\cite{CDM10} an infinite-dimensional Frobenius manifold structure was defined on an open subset $M_0$ of $M$ with the property that the associated flat pencil of metrics induces on $\cL M_0$ the Poisson pencil of the (dispersionless) 2D Toda hierarchy. For this reason we refer to $M_0$ as the 2D Toda Frobenius manifold. 

In this section we find the explicit solution to the flatness equations of the deformed connection of the 2D Toda Frobenius manifold. The expansion of such solution for the deformation parameter $\zeta \sim \infty$ defines a sequence of Hamiltonian densities on $M_0$. These in turn define, through the hydrodynamic type Poisson structure associated with the flat metric of the Frobenius manifold $M_0$, a hierarchy of commuting equations on the loop space $\cL M_0$, the so-called Principal hierarchy. We show that such hierarchy coincides with the extended 2D Toda hierarchy introduced in the previous section.

\subsection{The manifold $M_0$ as a bundle on the space of parametrized simple curves}

The manifold $M_0$ was defined in~\cite{CDM10} as the open subset of $M$ given by pairs of functions $(\la(z),\bla(z))\in M$ that satisfy the conditions
\begin{enumerate}
\item[i.] the coefficient $\bar u_{-1}$ in the expansion $\bla(z) = \bar u_{-1} z^{-1}+ O(1)$ for $z\to0$ is non-zero and the derivative of $w(z) := \la(z) + \bla(z)$ does not vanish on the unit circle $\S^1 := \{ z \in \C \text{ s.t. }  |z|=1\}$;
\item[ii.] the closed curve $\Gamma$ parametrized by the restriction of $w(z)$ to $\S^1$ is positively oriented, non-selfintersecting and encircles the origin $w=0$.
\end{enumerate}
Here we also require that
\begin{enumerate}
\item[iii.] the functions $\la(z)$, $\bla(z)$ are non-vanishing for $z$ in $D_\infty$, $D_0$ respectively.
\end{enumerate}

Condition (i.) guarantees the invertibility of the metric $\eta$ and condition (ii.) the solvability of the Riemann-Hilbert problem defining the flat coordinates, see Section~\ref{fl-sec}. Note that, by Lemma~\ref{windingzeros}, condition (iii.) implies that the winding conditions defining $M_1$ are satisfied, i.e. $M_0 \subset M_1$. 

The manifold $M_0$ can be seen as (an open subset of) a trivial two-dimensional fiber bundle over the space $M_\mathrm{red}$ of parametrized simple analytic curves, as shown by the map
\[
\isomorphism{M_0}{M_\mathrm{red}\oplus \C \oplus \C}{(\la(z),\bla(z))}{(w(z),v,u)}
\]
where $w(z):=\la(z)+\bla(z) \in \cH(\S^1)$, $v:=\bar u_0 = (\bla)_0$ and $e^u:=\bar u_{-1} = (\bla)_1$. Note that this map can be easily inverted by
\[
\la(z) = w_{\leq 0}(z) +z -v -e^u z^{-1}, \quad
\bla(z) = w_{\geq1}(z) -z +v + e^u z^{-1} .
\]
We refer to the variables $(w(z), v,u)$ as $w$-coordinates. 

In analogy with the construction of the tangent and cotangent spaces to $M$ given in the previous section, the $w$-coordinates suggest to identify both the tangent and the cotangent spaces at a point $\hat\la \in M_0$ with
\bean
T_{\hat\la}M \cong T_{\hat\la}^*M \cong \cH(\S^1)\oplus\C^2.
\eean
A vector $\pa_{X}$ is now represented by a triple $X=(X(z),X_v,X_u)$ where $X(z)=\pa_{X} w(z)$, $X_v=\pa_{X} v$ and $X_u=\pa_{X} u$. A $1$-form $\alpha$ is represented by a triple $(\alpha(z),\alpha_v,\alpha_u)$ through the pairing
\[
<\alpha, X> = \frac1{2 \pi i} \oint_{|z|=1} 
\alpha(z) X(z) \frac{dz}{z} + \alpha_v X_v + \alpha_u X_u .
\]
Switching to this representation is achieved by the following formulas: a vector $X=(X(z), \bar X(z)) \in \cH(D_\infty) \oplus \frac1z \cH(D_0)$ and a $1$-form  $\alpha=(\alpha(z), \bar\alpha(z) ) \in \cH(D_{0}) \oplus z\cH(D_{\infty})$ are represented by triples
\bean
&X = \big( X(z) + \bar X(z) , (\bar X(z))_0, \frac1{\bar u_{-1}} ( \bar X)_{-1} \big), \\
&\alpha= \big(\alpha(z) + (\bar\alpha(z))_{<0}, ( \bar\alpha(z)-\alpha(z))_0,
\bar u_{-1} ( \bar\alpha(z) -\alpha(z) )_1 \big),
\eean
in $\cH(\S^1)\oplus\C^2$.

This representation of the tangent and cotangent spaces turns out to be quite convenient and natural (see e.g. the simple formulas for the metric~\eqref{eta} and the connection~\eqref{christ}). In the following we will freely use both representations of vectors and covectors, often without specifying which one we are using, since in most cases it will be clear from the context.

\subsection{The metric}

In~\cite{CDM10} the metric, i.e. a bilinear form $\eta$ on the cotangent space $T^*M$, was defined in terms of a linear map $\eta^*:T^*M \to TM$ by 
\[
\eta(\alpha,\beta) = <\alpha, \eta^*(\beta)>
.\]
The map $\eta^*$ is defined as follows: a $1$-form $\alpha= (\alpha(z), \bar\alpha(z) )\in T^*M$ is mapped to a vector $X = \eta^*(\alpha)=(X(z), \bar X(z))$ given by
\bean
&X(z) =(z\la'(z)\alpha(z)+z\bla'(z)\bar\alpha(z))_{\leq0}-z\la'(z)(\alpha(z)-\bar\alpha(z))_{<0},\\
&\bar{X}(z) =(z\la'(z)\alpha(z)+z\bla'(z)\bar\alpha(z))_{>0}+z\bla'(z)(\alpha(z)-\bar\alpha(z))_{\geq0} . 
\eean
The map $\eta^*$ is invertible, i.e. the bilinear form $\eta$ is non-degenerate, at the points $(\la,\bla)$ of $M$ such that $w'(z) = \la'(z) + \bla'(z) \not= 0$ for $z \in \S^1$ and $\bar u_{-1} \not= 0$, i.e. in particular at  the points of $M_0$.

Representing vectors and $1$-forms by elements in $\cH(\S^1)\oplus\C^2$, the map $\eta^*$ is given by
\[
\Function{\eta^*}{T^*M}{TM}{(\alpha(z),\alpha_v,\alpha_u)}{(zw'(z) \alpha(z),\alpha_u,\alpha_v)}
\]
and its inverse is obviously
\[
\Function{\eta_*}{TM}{T^*M}{(X(z),X_v,X_u)}{(\frac{1}{zw'(z)}X(z),X_u,X_v)} .
\]
The bilinear form on the tangent space is written in $w$-coordinates as 
\beq
\label{eta-t}
\eta(X,Y)
=\frac{1}{2\pi i}\oint_{|z|=1} \frac{X(z)Y(z)}{z^2 w'(z)} dz + X_v Y_u + X_u Y_v .
\eeq
and on the cotangent space as
\beq
\label{eta}
\eta(\alpha,\beta)
=\frac{1}{2\pi i}\oint_{|z|=1} \alpha(z)\beta(z) w'(z) dz + \alpha_v\beta_u + \alpha_u\beta_v .
\eeq
Note that changing variable of integration this can in turn be expressed as an integral over the curve $\Gamma = w(\S^1)$ 
\[
\eta(\alpha,\beta)
=\frac{1}{2\pi i}\oint_{\Gamma} \alpha(z(w))\beta(z(w)) dw + \alpha_v\beta_u + \alpha_u\beta_v .
\]

\subsection{The Levi-Civita connection}

Now we derive a formula for the Levi-Civita connection of the metric $\eta$.
Let us define the Christoffel symbol $\Gamma$ of $\eta$ as a map
\[
\Gamma: TM \otimes T^*M \to T^*M
\]
that associates to $X \in TM$, $\alpha \in T^*M$ a $1$-form $\Gamma_X(\alpha) \in T^*M$; 
representing vectors and $1$-forms by elements in $\cH(\S^1)\oplus\C^2$, we define
\beq
\label{christ}
\Gamma_X(\alpha)=\left( \frac{\alpha'(z) X(z)}{w'(z)},0,0\right) .
\eeq
The covariant derivative of the $1$-form $\alpha$ along the vector field $X$ is defined by
\beq
\label{cov-der}
\nabla_X\alpha=\pa_X\alpha-\Gamma_X(\alpha) =( \pa_X \alpha(z)-\frac{\alpha'(z) X(z)}{w'(z)}, \pa_X\alpha_v, \pa_X\alpha_u) .
\eeq

\begin{proposition}
The connection $\nabla$ is torsion free and compatible with the metric $\eta$. 
\end{proposition}
\begin{proof}
The compatibility of $\nabla$ with the metric $\eta$ is equivalent to the identity
\[
\pa_X\left(\eta(\alpha,\beta)\right)= \eta(\nabla_X\alpha,\beta)+\eta(\alpha,\nabla_X\beta)
\]
for every $\alpha,\beta$ in $T^* M$ and $X$ in $T M$.
Let us first compute the following derivative along the vector $X$, using Leibniz rule
\[
\begin{split}
&\pa_X \oint \alpha(z)\beta(z) zw'(z) \frac{dz}{z}=\\
&= \oint \left[\left((\pa_X\alpha(z))\beta(z)+\alpha(z)(\pa_X\beta(z))\right)  zw'(z) 
-z\pa_z(\alpha(z)\beta(z)) X(z) \right]\frac{dz}{z}
\end{split}
\]
where in the last summand we used commutativity of $\pa_X$ and $\pa_z$ and integration by parts. The last expression is equal to
\[
\begin{split}
&\oint \left[\left((\pa_X\alpha(z))\beta(z)+\alpha(z)(\pa_X\beta(z))\right)  zw'(z) \right.\\
&\left. -(z\alpha'(z) X(z))\beta(z) -\alpha(z)(z\beta'(z) X(z))\right] \frac{dz}{z}= \\
&= \oint \left((\nabla_X\alpha)(z)\beta(z)+\alpha(z)(\nabla_X\beta)(z)\right)  zw'(z) \frac{dz}{z}.
\end{split}
\]
This formula allows us to easily take the derivative of formula~\eqref{eta}
\[
\begin{split}
\pa_X\left(\eta(\alpha,\beta)\right)&=\frac{1}{2\pi i}\oint \left((\nabla_X\alpha)(z)\beta(z)+\alpha(z)(\nabla_X\beta)(z)\right)  zw'(z) \frac{dz}{z}+\\
&+(\pa_X\alpha_v)\beta_u+(\pa_X\alpha_u)\beta_v+\alpha_v(\pa_X\beta_u)+\alpha_u(\pa_X\beta_v)=\\
&=\eta(\nabla_X\alpha,\beta)+\eta(\alpha,\nabla_X\beta) .
\end{split}
\]
The compatibility of $\nabla$ with the metric $\eta$ is proved.

The fact that the torsion of $\nabla$ is zero is equivalent to the identity
\beq
\label{torsionless}
<\Gamma_X(\alpha),Y>=<\Gamma_Y(\alpha),X>
\eeq
for every $\alpha,\beta$ in $T^*M$ and $X,Y$ in $TM$.
The proof is immediate. 
\end{proof}

\subsection{Flat coordinates}
\label{fl-sec}

Let us briefly recall the construction of the flat coordinates given in~\cite{CDM10} and obtain explicit formulas for the flat coordinates functionals. 

The simple curve $\Gamma = w(\S^1)$ divides the Riemann sphere in an interior and an exterior domain. We denote their closures by $\Gamma_0$ and $\Gamma_\infty$ respectively. The inverse $z(w)$ of the function $w(z)$ defines an holomorphic function on $\Gamma$, i.e. $z(w) \in \cH(\Gamma)$. 

Consider the following Riemann-Hilbert factorization problem: find two non-vanishing functions $f_0 \in \cH(\Gamma_0)$, $f_\infty \in w\cH(\Gamma_\infty)$ such that 
\[
z(w) = \frac{f_\infty(w)}{f_0(w)} \text{ for } w \in \Gamma
\]
and with normalization fixed by $f_\infty(w) = w +O(1)$ for $|w|\to\infty$. The solution to this factorization problem always exists and is unique. 

The coefficients $t^\alpha$ in the expansions 
\[
\log f_0(w) = - \sum_{\alpha\geq0} t^\alpha w^\alpha, \quad
\log \frac{f_\infty(w)}w = \sum_{\alpha<0} t^\alpha w^\alpha ,
\]
respectively in a neighborhood of $w=0$ and $\infty$, along with $t^u=u$ and $t^v=v$, form a system of flat coordinates. 

Indeed one can easily see that
\[
\frac{\pa w(z)}{\pa t^\alpha} = - z w'(z) w^\alpha 
\]
and, by substituting in~\eqref{eta-t}, one gets the nontrivial components of the Gram matrix $\eta\left(\frac{\pa}{\pa t^{\hat\alpha}} , \frac{\pa}{\pa t^{\hat\beta}} \right)$ of the metric in flat coordinates
\beq
\label{eta-up}
\eta_{\alpha \beta} = \de_{\alpha+\beta, -1}, \quad
\eta_{u v} =\eta_{vu}= 1. 
\eeq

Note that $\log\frac{z(w)}w \in \cH(\Gamma)$ is given by the sum
\[
\log\frac{z(w)}{w} = \log \frac{f_\infty(w)}w - \log f_0(w) . 
\]
For $\alpha\leq-1$ 
\[
\frac1{2 \pi i} \oint_\Gamma \log f_0(w) \ w^{-\alpha-1} \ dw =0
\]
hence
\[
\frac1{2 \pi i} \oint_\Gamma \log\frac{z(w)}{w} \ w^{-\alpha-1} \ dw = 
\frac1{2 \pi i} \oint_\Gamma \log \frac{f_\infty (w) }{w} \ w^{-\alpha-1} \ dw = t^\alpha ,
\]
where the last equality is evaluated by deforming the contour of integration to the neighborhood of $w=0$. By changing the variable of integration the left-hand side is written
\[
\frac1{2 \pi i} \oint_{|z|=1} \left( \log\frac{z}{w(z)} \right)' \ \frac{w^{-\alpha}(z)}{\alpha} \ dz = 
\frac1{2 \pi i} \oint_{|z|=1} \frac{w^{-\alpha}}{\alpha} \ \frac{dz}z
\]
A similar computation can be performed in the case $\alpha\geq0$. We obtain the following representation of the flat coordinates as integrals on the unit circle
\bean
&t^{\alpha} = \frac1{2 \pi i} \oint_{|z|=1} \frac{w^{-\alpha}}{\alpha}\frac{dz}z \text{ for } \alpha\not=0, \\
&t^{0} = -\frac1{2 \pi i} \oint_{|z|=1} \log\frac{w(z)}{z} \frac{dz}z .
\eean

Observe that by substituting the differentials
\beq
\label{diff-dt}
dt^\alpha = ( - w^{-\alpha-1}(z) , 0 ,0 ), \quad
dv = (0,1,0),\quad
du=(0,0,1)
\eeq
in formula~\eqref{cov-der} one can easily check that the $t^{\hat\alpha}$ are flat functions with respect to the Levi-Civita connection $\nabla$, i.e. 
\[
\nabla d t^{\hat\alpha} =0 .
\]

\subsection{The associative product and the deformed flat connection}

The Frobenius manifolds are endowed with an associative commutative product on each tangent space.  
In the case of the 2D Toda Frobenius manifold this product was defined in \cite{CDM10} by introducing a multiplication on the cotangent spaces and then dualizing it via $\eta$ to the tangent bundle. 
The product of two $1$-forms $\alpha=(\alpha(z), \bar\alpha(z) )$ and $\beta=(\beta(z), \bar\beta(z) ) $, represented by pairs of functions in $\cH(D_0) \oplus z \cH(D_\infty)$, is given by 
\bean
\alpha \cdot \beta=
&\left(\hspace{3pt}\alpha\left[z\la'\beta+z\bla'\bar\beta\right]_{>0}+\left[z\la'\alpha+z\bla'\bar\alpha\right]_{>0}\beta-\left[z\la'\alpha\beta+z\bla'(\alpha\bar\beta+\bar\alpha\beta)\right]_{\geq 0},\right.\\
&\left.-\bar\alpha\left[z\la'\beta+z\bla'\bar\beta\right]_{\leq0}-\left[z\la'\alpha+z\bla'\bar\alpha\right]_{\leq0}\bar\beta
+\left[z\bla'\bar\alpha\bar\beta+z\la'(\alpha\bar\beta+\bar\alpha\beta)\right]_{\leq 1}\right).
\eean


Given a vector field $X$ over $M_0$, the multiplication by $X$ induces a linear map $X\cdot :TM \to TM$ sending $Y\mapsto X\cdot Y$. This map can be dualized to the cotangent bundle, giving a linear map $C_X:T^*M \to T^*M$ that coincides with the multiplication by $\eta_*(X)$ on the cotangent space: $C_X(\alpha) = \eta_*(X) \cdot\alpha$, where $\alpha$ is a $1$-form. Now we obtain the explicit form of this operator, representing vectors and $1$-forms as elements of $\cH(\S^1)\oplus\C^2$.
\begin{lemma}
Let $X=(X(z), X_v, X_u)$ be a vector field on $M_0$. The operator $C_X$ is given by
\bea
\label{c-x}
&C_X(\alpha) = \\
&\Big( \frac{X(z)}{zw'(z)} \left( 
(z w'(z) \alpha(z) )_{>0} - (z w'(z))_{>0} \alpha(z) + z \alpha(z)
+\frac{e^u}z (\alpha(z) +\alpha_v) +\alpha_u \right) \nn\\
&+ \left(X_{>0}(z) \alpha(z)\right)_{<0} + 
\left( X_{\leq0}(z) \alpha(z)\right)_{\geq0} + \frac{e^u}z X_u (\alpha(z) +\alpha_v) +X_v \alpha(z) ,\nn\\
&\left( X(z) \alpha(z) \right)_0 +X_u \alpha_u +X_v \alpha_v,\nn\\
&\left( e^u (X(z) + zw'(z) X_u)(\alpha(z)+\alpha_v) \right)_1 
-e^u X_u \alpha_v + X_v \alpha_u \Big) \nn
\eea
where $\alpha=(\alpha(z), \alpha_v, \alpha_u)$ is a $1$-form on $M_0$.
\end{lemma}

The unit vector field is given by
\[
e=\frac{\pa}{\pa v}
\]
which can be represented as $e=(-1,1)$ or as $e=(0,1,0)$.

Recall~\cite{CDM10} that the Frobenius manifold $M_0$ is endowed also with an Euler vector field 
\[
E = (\la(z) - z \la'(z) ,\bla(z) -z \bla'(z))
\]
or, equivalently
\[
E = ( w(z) - z w'(z), v, 2).
\]

The deformed flat connection on $M_0\times\C^\times$ is defined as the deformation $\tilde\nabla$ of the Levi-Civita connection $\nabla$ (on the tangent bundle) given by the following formulas
\bean
&\tilde\nabla_X Y = \nabla_X Y + \zeta X \cdot Y ,\\
&\tilde\nabla_{\frac{d}{d\zeta}} Y = \pa_\zeta Y + E \cdot Y - \frac1\zeta \cV( Y ) ,
\eean
where $X$, $Y$ are vector fields on $M_0$ and $\zeta\in\C^\times$ is the deformation parameter. 
The remaining components of $\tilde\nabla$ are assumed to be trivial, i.e. $\tilde\nabla_X \frac{d}{d\zeta} = \tilde\nabla_{\frac{d}{d\zeta}} \frac{d}{d\zeta}=0$. 
The operator $\cV$ on $TM$ is defined in terms of the Euler vector field $E$ by 
\beq
\label{V-oper}
\cV = \frac12 - \nabla E 
\eeq
since the charge of the Frobenius manifold $M_0$ is $d=1$. 

Dualizing the definition of $\tilde\nabla$ to the cotangent bundle we obtain the following formulas for its nontrivial components 
\bean
&\tilde\nabla_X \alpha = \nabla_X \alpha - \zeta C_X(\alpha), \\
&\tilde\nabla_{\frac{d}{d\zeta}}\alpha = \pa_\zeta \alpha - \cU(\alpha) + \frac1\zeta \cV(\alpha),
\eean
where $\alpha$ is a $1$-form field, we denote by the same symbol $\cV$ the transpose of the operator defined in~\eqref{V-oper} and by $\cU = C_E$ the operator of multiplication by $\eta_*(E)$ on the contangent bundle. The invariance of the metric clearly implies that the operator $\cU$ is symmetric w.r.t. $\eta$.

\begin{lemma}
\label{lemma-V}
The operator $\cV: T^*M \to T^*M$ on $M_0$ is antisymmetric w.r.t. the metric $\eta$ and is explicitly given by
\[
\cV(\alpha) = \left( -\frac{\alpha(z)}2 - z \alpha'(z) \left( \frac{w(z)}{z w'(z)}\right), -\frac{\alpha_v}2 , \frac{\alpha_u}2 \right)
\]
where $1$-forms are represented by elements in $\cH(\S^1)\oplus\C^2$. The flat coordinates differentials $\{dt^{\hat\alpha}\}_{\hat\alpha\in\hat\Z}$ are eigenvectors of $\cV$, i.e. for $\alpha\in\Z$
\[
\cV(dt^\alpha) = (\alpha+\frac12) dt^\alpha, \quad
\cV(dv) = -\frac12 dv, \quad
\cV(du) = \frac12  du
.\]
\end{lemma}
\begin{proof}
Let $\alpha$ be a $1$-form and $X$, $Y$ vector fields. The covariant derivative of $Y$ is 
\[
\nabla_X Y = \pa_X + \Gamma^*_X(Y)
\]
where the transpose $\Gamma^*_X: TM \to TM$ of the Christoffel symbol~\eqref{christ} is given by
\[
\Gamma^*_X ( Y) = \left( -z \pa_z \left( \frac{X(z) Y(z)}{z w'(z)} \right), 0,0 \right) .
\]
We then compute 
\[
\nabla_X(E) = \left( X(z) - z \pa_z \left( X(z) \frac{w(z)}{z w'(z)} \right) , X_v ,0 \right) 
\]
and consequently the operator $\cV$ on the tangent bundle is given by
\[
\cV(X) = \frac{X}2 - \nabla_X(E) = \left( -\frac{X(z)}2 +z \pa_z \left( X(z) \frac{w(z)}{z w'(z)} \right), 
-\frac{X_v}2 , \frac{X_u}2 \right) .
\]
Transposing we obtain the desired expression for $\cV$ on the cotangent bundle. 

Using such expression, the antisymmetry of $\cV$, i.e.
\[
\eta(\alpha, \cV(\beta)) + \eta(\cV(\alpha),\beta) =0
\]
for any $1$-forms $\alpha$, $\beta$, can be easily checked.

The proof is completed by computing $\cV(dt^{\hat\alpha})$ using the explicit form of the differentials~\eqref{diff-dt}.
\end{proof}

\subsection{Deformed flat coordinates}

A functional $y(\hat\la,\zeta)$ on $M_0\times\C^\times$ is called {\it deformed flat} if its differential is horizontal w.r.t. the deformed flat connection $\tilde\nabla$, i.e.
\[
\tilde\nabla dy = 0 .
\]
The general theory of (finite dimensional) Frobenius manifolds ensures that the deformed connection $\tilde\nabla$ is flat. This in turn implies the local existence of a system of deformed flat coordinates i.e. a set of independent deformed flat functions. In this section we will provide an explicit system of deformed flat coordinates on $M_0$, proving in particular the flatness of $\tilde\nabla$ for the infinite-dimensional Frobenius manifold $M_0$.

The deformed flatness equations for a differential $dy$ are 
\bes
\label{horizz}
\bea
&\pa_X dy = \left(\Gamma_X  + \zeta C_X \right)(dy), \label{horizontality-X} \\
&\pa_\zeta dy = \left(\cU - \frac1\zeta \cV\right)(dy) ,
\label{horizontality-z}
\eea
\ees
for any vector field $X$ on $M_0$.

The following theorem provides an infinite family of deformed flat functions.
\begin{theorem}
\label{theorem2} 
The family of functionals $\left\{ y_{\hat\alpha}(\hat\la,\zeta)\right\}_{\hat\alpha\in\hat\Z}$ over $M_0\times\C^\times$ defined by
\bes
\label{y-functionals}
\bea
y_{\alpha}(\zeta) :=& -\frac{\zeta^{-1/2}}{2 \pi i} \oint_{|z|=1}  
\Big[\frac{(\l\zeta+\bl\zeta)^{\alpha+1}}{\alpha+1} \exp(\frac{\bl\zeta-\l\zeta}2)\Big] \ \frac{dz}z 
\quad \text{for} \quad \alpha\not= -1 ,\\
y_{-1}(\zeta) :=& -\frac{\zeta^{-1/2}}{2 \pi i} \oint_{|z|=1} 
\Big[e^{-\la \zeta} \Big( \log \frac{\la\zeta+\bla\zeta}{z} - \log (\frac{\la\zeta}z) + \ein(-\la\zeta) -1 \Big) +\nn \\
&\qquad \qquad+\exp(\frac{\bla\zeta-\la\zeta}{2})\Big] \frac{dz}z ,
\\
y_v(\zeta) :=&\ \frac{\zeta^{-1/2}}{2 \pi i} \oint_{|z|=1} 
\Big[-e^{- \la \zeta} \Big( \log (\frac{\la\zeta+\bla\zeta}{z}) - \log (\frac{\la\zeta}z) + \ein(-\la\zeta) -1 \Big)+ \nn \\
&\quad +  e^{\bla \zeta} \Big( \log  \frac{\la\zeta + \bla\zeta}{z} + \log ( z \bla\zeta ) - \ein(\bla \zeta) -1\Big)-2\log(\zeta)\Big] \frac{dz}{z},
\\
y_u(\zeta) :=&\ \frac{\zeta^{-1/2}}{2 \pi i} \oint_{|z|=1}
\big[\exp(\bla \zeta)-1\big]  \frac{dz}{z} .
\eea
\ees
forms a Levelt system of deformed flat coordinates on $M_0$ at $\zeta=0$.
\end{theorem}

The functions $y_{\hat\alpha}(\hat\la,\zeta)$ are multivalued in the variable $\zeta\in\C^\times$, with a branch point at $\zeta=0$. The branch point is of logarithmic type in the case of $y_v$ and of algebraic type in all the other cases. Note that the behavior of $y_v$ when $\zeta \to e^{2\pi i} \zeta$ is
\[
y_v \to e^{-\pi i} ( y_v + 4\pi i y_u ) . 
\]
The logarithmic branch point is a consequence of the presence of resonance in the spectrum of the Frobenius manifold, as we will see in the next subsection. 

Note that the entire exponential integral $\ein(z)$ is an entire function defined by the power series
\begin{equation} \label{einfn}
 \ein (z) := - \sum_{n=1}^{\infty} \frac{(-z)^n}{n!\ n} 
\end{equation}
with derivative $z \ein'(z) = 1- e^{-z}$.

In order to prove Theorem \ref{theorem2} we consider a general class of functionals of the form
\beq
\label{class_local_functionals}
y(\la,\bla;\zeta) := \frac{\zeta^{-\frac12}}{2 \pi i} \oint_{|z|=1} F(\zeta\la(z),\zeta\bla(z)) \frac{dz}{z} +\phi(\zeta)
\eeq
and give necessary and sufficient conditions for them to be deformed flat. In~\eqref{class_local_functionals} we assume that $F(x,\bx)$ is an analytic function on an open set $\Omega$ in $\C^2$ and $\phi(\zeta)$ to be a multivalued holomorphic function in $\C$ with a branch point of algebraic or logarithmic type at $\zeta=0$.

Note that in the following the function $F$ and its partial derivatives $F_\x, F_\bx, F_{\x\x},F_{\x\bx}$ and $F_{\bx\bx}$ are always implicitly evaluated in $\x=\zeta\la(z)$ and $\bx=\zeta\bla(z)$. Projections like $(\ )_{\geq0}$, $(\ )_{0}$, etc. are always taken with respect to the variable $z$. 

Let us consider first the horizontality equation in the direction of the deformation variable $\zeta$. The first important consequence of the choice of functionals of type~\eqref{class_local_functionals} is that it will directly follow from the other horizontality equations, as shown by the following proposition.
\begin{proposition}
Let $y(\hat\la;\zeta)$ be a functional of the form~\eqref{class_local_functionals} such that its differential is horizontal in the direction of the Euler vector field $E$, i.e.
\beq
\label{E_horizontal}
\nabla_E (dy(\zeta))= \zeta \cU(dy(\zeta))
.\eeq
Then the horizontality in the direction $\zeta$ follows
\beq
\label{zeta-horizontal}
\pa_\zeta dy(\zeta)=\left(\cU-\frac{1}{\zeta}\cV\right)(dy(\zeta))
.\eeq
\end{proposition}
\begin{proof}
For a functional of type~\eqref{class_local_functionals} it is easy to check that the following identity holds
\beq
\label{homo-y}
\left[\zeta\pa_\zeta  +\frac12 \right] y(\zeta) = E\left(y(\zeta)\right) ,
\eeq
up to an irrelevant function that depends only on $\zeta$.
Indeed the second term in the right-hand side of 
\[
E(y(\zeta))= 
\zeta^{1/2} \frac1{2\pi i}\oint ( F_x\la+F_\bx\bla ) \frac{dz}{z} -\zeta^{1/2}\frac1{2\pi i}\oint ( F_x\la'+F_\bx\bla' ) dz
\]
vanishes, since it is the integral of a total derivative in $z$, and the first term equals $\zeta \pa_\zeta y(\zeta) + \frac12 y(\zeta)$.

Differentiating equation~\eqref{homo-y}, we obtain
\beq
\label{homo-2}
\zeta \pa_\zeta dy(\zeta) + \frac12 dy(\zeta) = d(E(y(\zeta))) .
\eeq
Then we rewrite the right-hand side as (see below)
\beq
\label{homo-3}
d(E(y(\zeta))) = \nabla_E dy(\zeta) + \nabla_{dy(\zeta)} E 
\eeq
and eliminate $\nabla_E dy(\zeta)$ using~\eqref{E_horizontal}, obtaining from~\eqref{homo-2}
\[
\zeta \pa_\zeta dy(\zeta) =
\left( \zeta \cU -\frac12 + \nabla E \right) (dy(\zeta))
\]
which, by definition of $\cV$, is the required equation
~\eqref{zeta-horizontal}.

Equation~\eqref{homo-3} is proved by contracting its left-hand side with an arbitrary vector field $X$, obtaining
\bean
<d(E(y(\zeta))),X> &= \pa_X < dy(\zeta), E> \\
&= < \nabla_X dy(\zeta), E> + <dy(\zeta), \nabla_X E>
.\eean
The last line, spelling out the covariant derivative in its first term  and transposing the operator $\nabla E$, is equal to
\[
<\pa_X dy(\zeta) , E> - <\Gamma_X (dy(\zeta)), E> 
+<\nabla_{dy(\zeta)}E,X>
.\]
Here we used the symbol $\nabla_{dy(\zeta)}E$ to denote the traspose w.r.t. $<,>$ of the operator $X \to \nabla_X E$ on the tangent bundle, evaluated on the $1$-form $dy(\zeta)$.  
Finally by the torsionless property~\eqref{torsionless} of the Christoffel operator $\Gamma_X$ and the symmetry of second derivatives, i.e. $<\pa_X dy,E>=<\pa_E dy,X>$, we obtain
\[
<d(E(y(\zeta))),X> = < (\pa_E - \Gamma_E + \nabla E) (dy(\zeta)), X >
\]
which, by the arbitrarity of $X$, is equivalent to~\eqref{homo-3}.
\end{proof}

Let us now consider the horizontality equations in the direction of $M_0$. We need the following simple lemma.
\begin{lemma}
\label{z_derivation} 
Let $G(\x,\bx)$ be a function analytic in $\x$ and $\bx$.
The following identity holds
\[
z\pa_zG=\zeta\left( G_\x \left(zw'(z)\right)_{\leq0} + G_\bx \left( zw'(z)\right)_{>0} - (G_\bx-G_\x)(z+\frac{e^u}{z})\right),
\]
where $G=G(\zeta\la(z),\zeta\bla(z))$, $G_\x=\frac{\pa G}{\pa \x}(\zeta\la(z),\zeta\bla(z))$ and $G_\bx=\frac{\pa G}{\pa \bx}(\zeta\la(z),\zeta\bla(z))$. 
\end{lemma}

For functionals of type~\eqref{class_local_functionals} we can reformulate~\eqref{horizontality-X} in term of simpler constraints on the function $F(\x,\bx)$. 
\begin{proposition}
\label{horizontal_lemma}
For a functional $y(\hat\la,\zeta)$ of type~\eqref{class_local_functionals} the horizontality equation~\eqref{horizontality-X} is equivalent to 
\beqa
\label{F-horizontality}
&\left(F_{x\bx}-F_{xx}- F_x \right)_{\geq-1} = 0  \\
&\left(F_{\bx\bx}-F_{x\bx} - F_\bx \right)_{\leq1} = 0 \\
&(F_{\bx}-F_{\x} - F)_0 = c \\
&\pa_u \left[e^u (F_{\bx}-F_{\x} - F)\right]_1 = 0 
\eeqa
where $c$ is a constant.
\end{proposition}
\begin{proof}
The horizontality equation of the differential $dy(\zeta)$ w.r.t. the deformed flat connection,
\[
\tilde\nabla dy(\zeta)=0
,\]
can be split in three separate equations 
\[
\tilde\nabla_X dy(\zeta) = 0, \quad
\tilde\nabla_{\frac\pa{\pa v}} dy(\zeta) = 0, \quad
\tilde\nabla_{\frac\pa{\pa u}} dy(\zeta) =0
\]
where $X=(X(z),0,0)$, $X(z)\in\cH(\S^1)$, which we will consider separately.

{\it First part.} Here prove that the differential $dy(\zeta)$ of a functional of the form~\eqref{class_local_functionals} is covariantly constant along any vector of the form $X=(X(z),0,0)$ iff 
\beq \label{firstpart}
(F_{x \bx} - F_{xx} - F_x )_{\geq-1} = 0 \text{ and }
(F_{\bx \bx} - F_{x \bx} - F_\bx )_{\leq0} = 0
.\eeq

The differential of a functional of type~\eqref{class_local_functionals} at fixed $\zeta\in\C^\times$  is given by 
\[
 d y(\zeta)=(\frac{\pa y(\zeta)}{\pa w(z)},\frac{\pa y(\zeta)}{\pa v},\frac{\pa y(\zeta)}{\pa u}) \in \cH(\S^1)\oplus \C^2 
\]
where the components are defined by
\bes
\label{differential}
\bea
\frac{\pa y(\zeta)}{\pa w(z)}&=\sqrt{\zeta}\left(\left(F_x\right)_{\geq 0}+\left(F_{\bar x}\right)_{<0}\right), \\
\frac{\pa y(\zeta)}{\pa v}&=\sqrt{\zeta}\left( F_{{\bar x}} - F_x\right)_{0}, \\
\frac{\pa y(\zeta)}{\pa u}&=\sqrt{\zeta}\ e^u\left(F_{\bar x}-F_x\right)_{1}.
\eea
\ees

The $v$ and $u$ components of the horizontality equation 
\bea \label{hori1}
\pa_X dy(\zeta) = \Gamma_X (dy(\zeta)) + \zeta C_X(dy(\zeta))
\eea
are written explicitly as
\bes
\bea
&\pa_X \frac{\pa y(\zeta)}{\pa v} = \zeta \left( X(z) \frac{\pa y(\zeta)}{\pa w(z)} \right)_0, \label{der-v}\\
&\pa_X \frac{\pa y(\zeta)}{\pa u}=
\zeta e^u \left( X(z) \left( \frac{\pa y(\zeta)}{\pa w(z)} +  \frac{\pa y(\zeta)}{\pa v}  \right) \right)_1 \label{der-u}
\eea
\ees
respectively.

The derivations of the components of $dy(\zeta)$ along $X=(X(z),0,0)$ are easily computed 
\bes
\bea
&\pa_X \frac{\pa y(\zeta)}{\pa w(z)} = \zeta^{\frac32} \left( (F_{xx} X_{\leq0} + F_{x\bx} X_{\geq1} )_{\geq0} 
+(F_{x\bx}X_{\leq0} +F_{\bx\bx} X_{\geq1})_{\leq-1}\right), \label{dx-dw} \\
&\pa_X \frac{\pa y(\zeta)}{\pa v} = 
\zeta^{\frac32} \left( (F_{x \bx} - F_{xx}) X_{\leq 0} + (F_{\bx \bx} - F_{x \bx} ) X_{\geq1} \right)_0, \\
&\pa_X \frac{\pa y(\zeta)}{\pa u} = \zeta^{\frac32} e^u 
\left( (F_{x \bx} - F_{xx}) X_{\leq0} + (F_{\bx \bx} - F_{x \bx}) X_{\geq1} \right)_1 .
\eea
\ees
Substituting in~\eqref{der-v} we obtain
\[
( (F_{x \bx} - F_{xx} -F_x) X_{\leq0} + (F_{\bx \bx} - F_{x \bx} - F_{\bx} ) X_{\geq1} )_0 = 0
\]
which is equivalent to
\beq
\label{FFF1}
(F_{x \bx} - F_{xx} -F_x)_{\geq0} = 0 \text{ and } 
(F_{\bx \bx} - F_{x \bx} - F_{\bx} )_{\leq-1} =0 .
\eeq
Similarly~\eqref{der-u} is equivalent to 
\beq
\label{FFF2}
(F_{x \bx} - F_{xx} -F_x)_{\geq1} = 0 \text{ and } 
(F_{\bx \bx} - F_{x \bx} - F_{\bx} )_{\leq0} =0.
\eeq

Taking into account these equations we can write $\frac{\pa y(\zeta)}{\pa w(z)}$ in terms of the second derivatives of $F$ only
\beq
\label{dy-subs}
\frac{\pa y(\zeta)}{\pa w(z)} = \sqrt\zeta \left( (F_{x\bx} -F_{xx})_{\geq0} + ( F_{\bx\bx}-F_{x\bx})_{\leq-1} \right). 
\eeq

The remaining component of the horizontality equation~\eqref{hori1}, taking into account the Christoffel~\eqref{christ} and the multiplication~\eqref{c-x} operators on the cotangent, is 
\beqa 
\label{der-w}
&\pa_X \frac{\pa y(\zeta)}{\pa w(z)} = \frac{X(z)}{z w'(z)} z\pa_z \left( \frac{\pa y(\zeta)}{\pa w(z)} \right)  + \\
&+ \zeta \frac{X(z)}{zw'(z)} \left[ 
\left(z w'(z) \frac{\pa y(\zeta)}{\pa w(z)} \right)_{>0} - (zw'(z))_{>0} \frac{\pa y(\zeta)}{\pa w(z)} +\right.\\
&\left.+\left( z+ \frac{e^u}z \right) \frac{\pa y(\zeta)}{\pa w(z)} +
\frac{e^u}z \frac{\pa y(\zeta)}{\pa v} + \frac{\pa y(\zeta)}{\pa u}\right]+\\
&+ \zeta \left( X_{>0}(z) \frac{\pa y(\zeta)}{\pa w(z)} \right)_{<0} 
-\zeta \left( X_{\leq0}(z) \frac{\pa y(\zeta)}{\pa w(z)}\right)_{\geq0} .
\eeqa

In this equation, the last line, substituting~\eqref{dy-subs} and reorganizing its terms, is equal to the left-hand side, as given by~\eqref{dx-dw}, plus the term
\[
- \zeta^{\frac32} X(z) F_{x \bx} .
\]

Hence~\eqref{der-w} reduces to an equation which does not depend on $X(z)$, i.e.
\beqa
\label{der-w1}
\zeta^{\frac32} z w'(z) F_{x \bx} &=  z\pa_z \left( \frac{\pa y(\zeta)}{\pa w(z)} \right)+\\
 +\zeta &\left[ 
\left(z w'(z) \frac{\pa y(\zeta)}{\pa w(z)} \right)_{>0} - (zw'(z))_{>0} \frac{\pa y(\zeta)}{\pa w(z)} +\right.\\
& \left. +\left( z+ \frac{e^u}z \right) \frac{\pa y(\zeta)}{\pa w(z)} +
\frac{e^u}z \frac{\pa y(\zeta)}{\pa v} + \frac{\pa y(\zeta)}{\pa u}\right] .
\eeqa
Using Lemma~\ref{z_derivation} we compute the derivative 
\bean
z &\pa_z \frac{\pa y(\zeta)}{\pa w(z)} =\\ & =\zeta^{\frac32}  \left( \left[(F_{xx}-F_{\x\bx}) \left[zw'(z)\right]_{\leq0} + F_{x\bx}\,zw'(z) - (F_{x\bx}-F_{xx})(z+\frac{e^u}{z})\right]_{\geq0}\right.\\
&+\left.\left[F_{x\bx}\,zw'(z) + (F_{\bx\bx}-F_{x\bx}) \left[zw'(z)\right]_{>0} - (F_{\bx\bx}-F_{x\bx})(z+\frac{e^u}{z})\right]_{<0}\right)\nn .
\eean
Using~\eqref{FFF1}-\eqref{FFF2} this may be rewritten as
\bean
z\pa_z  \frac{\pa y(\zeta)}{\pa w(z)}  =\ \zeta^{3/2} &\left( \left[-F_{\x} \left[zw'(z)\right]_{\leq0} - (F_{\x\bx}-F_{\x\x})z\right]_{\geq0}+\right.\\
&+ \left[F_{\bx} \left[zw'(z)\right]_{>0} - (F_{\bx\bx}-F_{\x\bx})z\right]_{<0}+\\
&+ F_{\x\bx}\,zw'(z) -\left.\frac{e^u}{z}\left(\left[F_\x\right]_{>0}+\left[F_\bx\right]_{\leq0}\right)\right) .
\eean
Using this and~\eqref{differential}, observing that
\bean
&-\left[F_{\x} \left[zw'(z)\right]_{\leq0} \right]_{\geq0}+\left[F_{\bx} \left[zw'(z)\right]_{>0}\right]_{<0}
+\left[zw'(z)\left(\left[F_\x\right]_{\geq 0}+\left[F_\bx\right]_{<0}\right)\right]_{>0}\\
&-\left[z w'(z) \right]_{>0} \left(\left[F_\x\right]_{\geq 0}+\left[F_\bx\right]_{<0}\right)= -\left[F_x \left[z w'(z) \right]_{\leq0} +F_\bx \left[z w'(z) \right]_{>0}\right]_0 ,
\eean
and taking into account~\eqref{FFF1}-\eqref{FFF2}, we finally rewrite~\eqref{der-w1} as
\bean
-&\left[F_{\x}\left[z w'(z)\right]_{\leq0}+ F_{\bx}\left[z w'(z)\right]_{>0} - (F_\bx-F_\x)(z+\frac{e^u}{z})\right]_0 - \\
-&\left[F_{\x\bx}-F_{\x\x}-F_\x \right]_{-1} =0 .
\eean
Due to Lemma \ref{z_derivation}, the first parenthesis is precisely $-\oint z\pa_z F \frac{d z}{z}=0$. We have completed the proof of~\eqref{firstpart}.

{\it Second part.} 
Now we show that $d y(\zeta)$ is covariantly constant along the vector $\frac{\pa}{\pa u}$ iff
\bes
\label{secondpart}
\bea
&(F_{x \bx} - F_{xx} -F_x)_{\geq1} = 0, \\
&(F_{\bx \bx} - F_{x \bx} - F_{\bx} )_{\leq1} =0, \\
&\pa_u \left[\frac{e^u}{z}\left(F_\bx-F_\x-F\right)\right]_{0}=0.
\eea
\ees
The components of the horizontality equation
\[
\frac{\pa}{\pa u} dy(\zeta) = \zeta C_u(dy(\zeta)) 
\]
are given by
\bes
\bea
&\frac{\pa}{\pa u} \frac{\pa y(\zeta)}{\pa w(z)} = \zeta \frac{e^u}{z} \left( \frac{\pa y(\zeta)}{\pa w(z)}+\frac{\pa y(\zeta)}{\pa v}\right), \label{uder-z}\\
&\frac{\pa}{\pa u} \frac{\pa y(\zeta)}{\pa v} = \zeta \frac{\pa y(\zeta)}{\pa u}, \label{uder-v} \\
&\frac{\pa}{\pa u} \frac{\pa y(\zeta)}{\pa u} = \zeta e^u \left( \left(w'(z) \left( \frac{\pa y(\zeta)}{\pa w(z)} +\frac{\pa y(\zeta)}{\pa v} \right) \right)_0 - \frac{\pa y(\zeta)}{\pa v} \right).  \label{uder-u}
\eea
\ees

The derivatives of the components of $dy(\zeta)$ are 
\bes
\bea
&\frac{\pa}{\pa u} \frac{\pa y(\zeta)}{\pa w(z)}=
\zeta^{3/2}\left(\left[\frac{e^u}{z}\left(F_{\x\bx}-F_{\x\x}\right)\right]_{\geq 0}+\left[\frac{e^u}{z}\left(F_{\bx\bx}-F_{\x\bx}\right)\right]_{<0}\right), \\
&\frac{\pa}{\pa u} \frac{\pa y(\zeta)}{\pa v}=
\zeta^{3/2}\left[\frac{e^u}{z}\left(F_{\bx\bx}-2F_{\x\bx}+F_{\x\x}\right)\right]_{0}, \\
&\frac{\pa}{\pa u} \frac{\pa y(\zeta)}{\pa u}=
\zeta^{1/2}\left[\frac{e^u}{z}\left(F_\bx-F_\x\right)\right]_{0}+
\zeta^{3/2}\left[\frac{e^{2u}}{z^2}\left(F_{\bx\bx}-2F_{\x\bx}+F_{\x\x}\right)\right]_{0}. \label{derivative_u,u}
\eea
\ees

By a simple computation we obtain that equations~\eqref{uder-z}-\eqref{uder-v} are equivalent to 
\[
(F_{x \bx} - F_{xx} -F_x)_{\geq1} = 0 \text{ and } 
(F_{\bx \bx} - F_{x \bx} - F_{\bx} )_{\leq1} =0.
\]

Equation~\eqref{uder-u} is explicity written as
\beqa
\label{expl-uu}
&\frac1\zeta\left[\frac{e^u}{z}\left(F_\bx-F_\x\right)\right]_{0}+
\left[\frac{e^{2u}}{z^2}\left(F_{\bx\bx}-2F_{\x\bx}+F_{\x\x}\right)\right]_{0}=\\
=&\left[ \frac{e^u}{z} zw'(z)(\left[F_\x\right]_{>0}+\left[F_\bx\right]_{\leq0})\right]_0\!\!\! -e^u \left[F_\bx-F_\x\right]_0 .
\eeqa
Observe now that 
\bean
&\left[ \frac{e^u}{z} zw'(z)(\left[F_\x\right]_{>0}+\left[F_\bx\right]_{\leq0})\right]_0 =\\
=&\left[ \frac{e^u}{z}\left( F_\x\left[z w'(z)\right]_{\leq0}+ F_\bx\left[z w'(z)\right]_{<0})\right)\right]_0
\eean
and, since
\bean
0=\frac1\zeta\oint z\pa_z \frac{e^u}{z}F \frac{\de z}{z}&= 
\left[\frac{e^u}{z}\left(F_\x\left[z w'(z)\right]_{\leq0} + F_\bx\left[z w'(z)\right]_{<0})\right)\right]_0 - \\
&-\frac1\zeta\left[ \frac{e^u}{z}F \right]_0
 -e^u\left[F_\bx-F_x\right]_0 -\left[\frac{e^{2u}}{z^2}(F_\bx-F_x)\right]_0 ,
\eean
we can rewrite equation~\eqref{expl-uu} as
\bean
&\frac1\zeta\left[\frac{e^u}{z}\left(F_\bx-F_\x\right)\right]_{0}+\left[\frac{e^{2u}}{z^2}\left(F_{\bx\bx}-2F_{\x\bx}+F_{\x\x}\right)\right]_{0}=\\
&=\frac1\zeta\left[ \frac{e^u}{z}F \right]_0+\left[\frac{e^{2u}}{z^2}(F_\bx-F_x)\right]_0 .
\eean
To conclude, observe that both sides of this equation are derivatives in $u$
\bean
&\frac1\zeta\left[ \frac{e^u}{z}F \right]_0+\left[\frac{e^{2u}}{z^2}(F_\bx-F_x)\right]_0=\pa_u \frac1\zeta\left[ \frac{e^u}{z}F \right]_0 , \\
&\frac1\zeta\left[\frac{e^u}{z}\left(F_\bx-F_\x\right)\right]_{0}+\left[\frac{e^{2u}}{z^2}\left(F_{\bx\bx}-2F_{\x\bx}+F_{\x\x}\right)\right]_{0}=\pa_u \frac1\zeta\left[\frac{e^u}{z}\left(F_\bx-F_\x\right)\right]_{0} ,
\eean
hence it can be written as
\[
\pa_u \left[\frac{e^u}{z}\left(F_\bx-F_\x-F\right)\right]_{0}=0 .
\]
This ends the proof of~\eqref{secondpart}.

{\it Third part.}
Finally observe that, since $C_{\frac\pa{\pa v}}$ is the identity map and the Christoffel operator $\Gamma_{\frac\pa{\pa v}}$ is zero, the horizontality equation along the vector field $\frac\pa{\pa v}$ is simply
\[
\frac{\pa}{\pa v} dy(\zeta) = \zeta dy(\zeta) .
\]
Integrating we obtain
\[
 \pa_{v} y(\zeta)=\zeta y(\zeta) + k(\zeta),
\]
where $k(\zeta)$ is a constant depending only on $\zeta$. Hence this horizontality equation is equivalent to 
\beq
\label{thirdpart}
(F_\bx-F_x -F)_0 = c
\eeq
for a constant $c$.

Combining the formulas~\eqref{firstpart}, \eqref{secondpart}, \eqref{thirdpart} we obtain the desired result. 
\end{proof}

\begin{proof}[Proof of Theorem \ref{theorem2}]
It is now easy to see that the $y_{\hat\alpha}(\zeta)$ are deformed flat. The fact that they form a Levelt system, i.e. that the corresponding fundamental matrix can be written in a certain normal form, will be shown in the next section. 

The functionals~\eqref{y-functionals} are of the form~\eqref{class_local_functionals} 
\[
y_{\hat\alpha}(\hat\la,\zeta)=\frac{\zeta^{-\frac12}}{2\pi i}\oint F_{\hat\alpha}(\zeta\la(z),\zeta\bla(z)) \frac{dz}{z} +\phi_{\hat\alpha}(\zeta) \quad \hat\alpha\in\hat\Z 
\]
where the functions $F_{\hat\alpha}(\x,\bx)$ and $\phi_{\hat\alpha}(\zeta)$ are given by
\bean
F_\alpha(\x,\bx)=&-\frac{(\x+\bx)^{(\alpha+1)}}{\alpha+1}\exp(\frac{\bx-\x}{2}) \qquad \mbox{for $\alpha\neq -1$},\\
F_{-1}(\x,\bx)=&-\exp(-\x)(\log(\frac{\x+\bx}{\x})+\ein(-\x)-1)-\exp(\frac{\bx-\x}{2}),\\
F_{v}(\x,\bx)=&-\exp(-\x)(\log(\frac{\x+\bx}{\x})+\ein(-\x)-1)+\\
	      &+\exp(\bx)(\log((\x+\bx)\bx)-\ein(\bx)-1),\nonumber\\ 
F_{u}(\x,\bx)=&\exp(\bx)-1,
\eean
and
\[
\phi_{v}(\zeta)=-2\log(\zeta), \quad \phi_{u}(\zeta)=0, \quad \phi_{\alpha}(\zeta)=0 \qquad \mbox{for $\alpha\in\Z$}.
\]
These functions satisfy the following identities
\bean
& F_{\alpha;\bx}-F_{\alpha;\x}-F_\alpha =-\frac{\delta_{\alpha,-1}}{x} \quad \text{for } \alpha\in\Z, \\
&F_{v;\bx}-F_{v;\x}-F_{v} = \frac{1}{\bx}-\frac{1}{\x}, \\
&F_{u;\bx}-F_{u;\x}-F_{u} = 1 .
\eean
It follows that equations~\eqref{F-horizontality} are satisfied, hence, by Proposition~\ref{horizontal_lemma}, the $y_{\hat\alpha}(\zeta)$ are deformed flat. 
\end{proof}

\subsection{Levelt basis, monodromy and orthogonality.}

We now show that the deformed flat functions~\eqref{y-functionals} actually form a Levelt basis of deformed flat coordinates. 

It is convenient to rewrite the system~\eqref{horizontality-z} on the tangent space by introducing the gradient $\nabla y := \eta^*(dy)$ of a functional $y$. Taking into account the symmetry of $\cU$ and the antisymmetry of $\cV$ we get
\beq
\label{hor-tang}
\pa_\zeta \nabla y = (\cU + \frac1\zeta \cV ) (\nabla y) .
\eeq
The operator $\cV$, defined in~\eqref{V-oper}, is diagonal in the basis $\nabla t^{\hat\alpha}$ of $T M_0$
\[
\cV \nabla t^{\hat\alpha} = \mu_{\hat\alpha} \nabla t^{\hat\alpha}
\]
for $\mu_\alpha= -\alpha-\frac12$, $\mu_v = \frac 12 = -\mu_u$.
Here $\cU = E \cdot$ denotes the multiplication by $E$ on the tangent bundle $T M_0$. 

Equation~\eqref{hor-tang} is an operator-valued linear system on the complex plane with a regular singularity at $\zeta=0$ and an irregular singularity at $\zeta=\infty$, depending on the point $\hat\la\in M_0$.

In analogy with the finite-dimensional case we can define the ``fundamental matrix'' $Y:TM_0 \to TM_0$ as the linear operator determined by 
\[
Y(\nabla t^{\hat\alpha}) = \nabla y^{\hat\alpha}
\]
where $y_{\hat\alpha}$ are the deformed flat functions defined in~\eqref{y-functionals}.
Clearly $Y$ depends on $\zeta$ and on the point $\hat\la\in M_0$. The index $\hat\alpha$ is raised by the metric in flat coordinates~\eqref{eta-up}. 

The fundamental matrix satisfies the equation
\beq \label{hor-tang-Y}
\pa_\zeta Y = (\cU + \frac1\zeta \cV ) Y,
\eeq
where the composition of operators on $T M_0$ is understood. 

Let $R: T M_0 \to T M_0$ be a symmetric nilpotent operator defined by
\beq
\label{RRR}
R\left(\frac{\pa}{\pa t^\alpha}\right) = R\left(\frac{\pa}{\pa u}\right) = 0, \quad
R\left(\frac{\pa}{\pa v}\right)=  2 \frac{\pa}{\pa u}.
\eeq

\begin{proposition}
The fundamental matrix $Y$ can be factorized as 
\beq
\label{normalform}
Y = \Theta\ \zeta^\cV \zeta^R 
\eeq
where $\Theta: T M_0 \to T M_0$ is the linear operator defined by 
\[
\Theta( \nabla t^{\hat\alpha}) = \nabla \theta^{\hat\alpha}
\]
which is analytic at $\zeta=0$ and has leading term 
\beq
\label{leadingterm}
\Theta_{|\zeta=0} \equiv \mathrm{Id}.
\eeq
\end{proposition}
\begin{proof}
The functionals $y_{\hat\alpha}$ are related to the $\theta_{\hat\alpha}$, analytic in $\zeta$, by the formulas~\eqref{yth}.
Taking the gradient and raising the indices in~\eqref{yth} we get 
\bean
&\nabla y^{\alpha}  = \zeta^{-\alpha - \frac12}\ \nabla\theta^\alpha, \\
&\nabla y^v = \zeta^{\frac12}\ \nabla\theta^v, \\
&\nabla y^u = \zeta^{-\frac12}\ \nabla\theta^u + 2 \zeta^{\frac12} \log\zeta \ \nabla\theta^v. 
\eean
We obtain exactly these expressions if we evaluate~\eqref{normalform} on $\nabla t^{\hat\alpha}$, taking into account that
\bean
&\zeta^\cV \zeta^R \nabla t^\alpha = \zeta^{-\alpha-\frac12} \nabla t^\alpha, \\
&\zeta^\cV \zeta^R \nabla v = \zeta^{\frac12} \nabla v,\\
&\zeta^\cV \zeta^R \nabla u = \zeta^{-\frac12} \nabla u +  2 \zeta^{\frac12} \log\zeta \ \nabla v.
\eean

The functionals $\theta_{\hat\alpha}$ at $\zeta=0$ coincide with the flat coordinates
\[
(\theta_{\hat\alpha})_{|\zeta=0} = t_{\hat\alpha},
\]
hence, taking the gradients at $\zeta=0$, we obtain~\eqref{leadingterm}.
\end{proof}

Let us now comment on the normal form~\eqref{normalform} in relation  with the usual theory of matrix-valued rational linear equations on the complex plane. 
In analogy with the finite-dimensional case, we say that a system of the form~\eqref{hor-tang-Y} is resonant if two or more eigenvalues of $\cV$ differ by a non-zero integer. In this sense our case is highly resonant, since all eigenvalues $\mu_{\hat\alpha}$ differ by non-zero integers, except for $\mu_v = \mu_{-1}=\frac12$ and $\mu_u =\mu_0=-\frac12$.

In finite dimensions, assuming $\cV$ is diagonalizable and non-resonant, the normal form of the fundamental matrix in a neighborhood of $\zeta=0$ is
\[
Y=\Theta(\zeta) \zeta^\cV, 
\]
with $\Theta$ uniquely determined by fixing $\Theta(0) = \mathrm{Id}$. In the resonant case the normal form of the fundamental matrix is~\eqref{normalform}, where one must allow for a nilpotent matrix $R=R_1+R_2+\dots$, with $R_{2n+1}$ symmetric, $R_{2n}$ skewsymmetric, and such that
\[
\zeta^{\cV}  R_k \zeta^{-\cV} = \zeta^k R_k
\]
for $k=1,2,\dots$

It is easy to check that the operator $R=R_1$ defined in~\eqref{RRR} satisfies these requirements. Therefore we can conclude that the system~\eqref{hor-tang-Y} admits a normal form that is completely analogous to the Levelt normal form constructed in the finite dimensional case. Correspondingly we say that $y_{\hat\alpha}$ is a Levelt system of deformed flat coordinates. 

Note that the resonance of the system implies that there is a residual gauge freedom in the choice of $R$ and $\Theta$ that we will exploit below. 

The vector space $V:=T_{\hat\la} M_0$, at a fixed point $\hat\la\in M_0$ together with the bilinear form $\eta$ and the operators $R$, $\cV$ on $V$ define the {\it spectrum} (or monodromy at $\zeta=0$) of the Frobenius manifold $M_0$.

As expected from the general theory of Frobenius manifolds, the monodromy does not depend on the point $\hat\la\in M_0$, as one can see from the fact that the operators $\cV$ and $R$ are constant in flat coordinates. This is indeed a reflection of the general property of isomonodromicity of the system~\eqref{horizz}.

\begin{remark}
Note that despite the high degree of resonance of $\cV$, in our case the matrix $R$ is very simple. This type of monodromy, where $R=R_1$, is typical of Frobenius manifolds originating from quantum cohomology. The potential of Frobenius manifold $M_0$ can indeed be written (see~\cite{CDM10})
\bean
F &= \frac1{4\pi i} \oint_\Gamma \oint_\Gamma \mathrm{Li}_3 \, \frac{\tilde{z}(w_1)}{\tilde{z}(w_2)} \ dw_1dw_2 + \frac1{2 \pi i} \oint_{\Gamma} \left( -e^{t^0} \tilde{z}(w) +\frac{e^s}{\tilde{z}(w)} \right) dw \\
&- e^{s+t^0} + (v+\frac{t^{-1}}2) \frac1{4 \pi i} \oint_\Gamma (t^0 + \log\frac{\tilde{z}(w)}{w} )^2 dw + \frac12 v^2 (s+t^0) 
\eean
where we have used a slightly different set of flat coordinates $(t^\alpha,v, s)$ where $s:=u-t^0$ and $\tilde{z}(w) := z(w)e^{-t^0}$. In this form the potential is given by a cubic part plus a deformation with possibly exponential dependence on the variables $s$ and $t^0$. The deformation is killed by performing the limit $s,t^0 \to - \infty$ and sending the remaining variables to $0$. This is called {\it point of classical limit} of the Frobenius manifold and, in the quantum cohomology case,	 corresponds to the point where the quantum cup product coincides with the ordinary cup product in cohomology. 
The structure constants at the point of classical limit are, in flat coordinates
\bean
&c^\alpha_{\beta\gamma} = \delta_{\beta+\gamma-\alpha, -1} (H_{-\beta-1} + H_{-\gamma-1}- H_{-\alpha-1}), \\
&c^u_{\alpha\beta} =\de_{\alpha,\beta} , \quad c^{\hat\alpha}_{v \hat\beta} = \de_{\hat\alpha, \hat\beta}
\eean
and zero otherwise. Here $H_n=1$ for $n\geq0$, $H_n=0$ otherwise. 
As proved in~\cite{D99}, at the point of classical limit $\hat\la_{\mathrm{class}}\in M_0$ the system~\eqref{hor-tang-Y} is automatically in normal form since 
\[
\lim_{\ \ \ \ \ \hat\la\to\hat\la_{\mathrm{class}}} \cU = R =R_1
\]
and by isomonodromicity it determines the spectrum of the Frobenius manifold. In the case of a quantum cohomology the operator $R$ corresponds to the multiplication by the Chern class $c_1(X)$ in ordinary cohomology of $X$. It would be interesting to understand if the Frobenius manifold $M_0$ admits a (quantum) cohomological origin as suggested by these observations. 
\end{remark}

%

\begin{remark} \label{ort-rem}
It was shown in~\cite{D99} that the Levelt fundamental matrix~\eqref{normalform} can be chosen in such a way that the analytic part $\Theta$ satisfies the orthogonality condition
\beq \label{orthogonnn} 
\Theta^* (-\zeta) \Theta (\zeta) \equiv 1.
\eeq
This condition is not satisfied by our choice of fundamental solution $Y$, since we have preferred to keep a simpler form for the deformed flat coordinates~\eqref{y-functionals}. 
However it is possible to obtain an orthogonal fundamental matrix $\tilde Y$ by a simple modification of $Y$. Note that, while in the nonresonant case the orthogonality condition follows from the symmetry properties of $\cU$ and $\cV$ in~\eqref{hor-tang-Y}, in the presence of resonances it must be imposed as an external condition on $\Theta$.

Let us first define the following functions $\tilde\theta$ on $M_0$ analytic in the parameter $\zeta$ in a neighborhood of $\zeta=0$
\bes
\label{odfc}
\bea
&\tilde\theta_\alpha(\zeta) = -  \frac{(2\alpha)!!}{2\pi i} \oint_{|z|=1} 
\Big[ \frac{e^{\frac{\la+\bla}2 \zeta}}{\zeta^{\alpha+1}} + \frac{e^{-\frac{\la+\bla}2 \zeta}}{(-\zeta)^{\alpha+1}} \Big]_{(+)} e^{\frac{\bla-\la}2 \zeta} \ \frac{dz}z \text{ for }\alpha\geq0, \label{oo1}\\
&\tilde\theta_{-1}(\zeta) =- \frac1{2\pi i} \oint e^{-\la \zeta} \left( \log\left(1+\frac{\bla}{\la} \right) + \ein(-\la\zeta) -\ein\left(-\frac{\la+\bla}2 \zeta\right) \right)\ \frac{dz}z, \\
&\tilde\theta_\alpha(\zeta) = \theta_{\alpha}(\zeta)=-\frac1{2\pi i} \oint_{|z|=1} \frac{(\la+\bla)^{\alpha+1}}{\alpha+1} e^{\frac{\bla-\la}2 \zeta} \ \frac{dz}z \text{ for }\alpha\leq-2, \\
&\tilde\theta_v(\zeta) = \frac1{2\pi i} \oint \left[ 
-e^{-\la \zeta} \left( \log\left( 1+\frac{\bla}{\la} \right) + \ein(-\la\zeta) -\ein\left(-\frac{\la+\bla}2 \zeta\right) \right) \right.+\\
&\quad\quad\qquad\qquad +
\left. e^{\bla \zeta} \left( \log( \bla(\la+\bla)) -\ein(\bla\zeta) -\ein\left(\frac{\la+\bla}2 \zeta\right)
\right)
\right] \frac{dz}z, \\
&\tilde\theta_u(\zeta) = \theta_u(\zeta) = \frac1{2\pi i} \oint_{|z|=1} \frac{e^{\bla\zeta}-1}{\zeta} \ \frac{dz}z,
\eea
\ees
and correspondingly
\bean
&\tilde y_{\alpha}(\zeta)  = \zeta^{\alpha + \frac12}\ \tilde\theta_\alpha(\zeta), \\
&\tilde y_v(\zeta) = \zeta^{-\frac12}\ \tilde\theta_v(\zeta) + 2 \zeta^{\frac12} \log \zeta \ \tilde\theta_u(\zeta), \\
&\tilde y_u(\zeta) = \zeta^{\frac12}\ \tilde\theta_u(\zeta) .
\eean
The bracket $[ \ ]_{(+)}$ in~\eqref{oo1} denotes the projection to non-negative powers of $\zeta$. Explicitly one has
\[
\Big[ \frac{e^{\frac{\la+\bla}2 \zeta}}{\zeta^{\alpha+1}} + \frac{e^{-\frac{\la+\bla}2 \zeta}}{(-\zeta)^{\alpha+1}} \Big]_{(+)} = 2 \sum_{n\geq0} \frac{\zeta^{2n}}{(2n+\alpha+1)!}\left(\frac{\la+\bla}2\right)^{2n+\alpha+1}.
\]

The corresponding fundamental matrix $\tilde Y$ and its analytic part $\tilde\Theta$, which are operators on $T M_0$ defined by
\[
\tilde Y ( \nabla t^{\hat\alpha} ) = \nabla \tilde y^{\hat\alpha},
\quad
\tilde \Theta( \nabla t^{\hat\alpha} ) = \nabla \tilde\theta^{\hat\alpha},
\]
are related as before by
\[
\tilde Y = \tilde\Theta\, \zeta^{\cV} \zeta^R .
\]

One can check that the fundamental matrix $\tilde Y$ is obtained from $Y$ by the right-composition with a constant invertible operator $C$ on $T M_0$
\[
\tilde Y = Y \, C .
\]
In components, where $C (\nabla t_{\hat\alpha}) = \nabla t_{\hat\gamma} \, C^{\hat\gamma}_{\ \hat\alpha}$, it is given by
\bean
&C^{\alpha}_{\ \beta}=\frac{(2\beta)!!}{(2\alpha)!!} \sum_{n\geq0} \de_{\alpha,\beta+2n} \text{ for } \alpha,\beta\geq0,\\
&C^{\alpha}_{\ -1} = (-1)^\alpha \frac{c_{\alpha+1}-1}{(2\alpha)!!} \text{ for } \alpha\geq0, \\
&C^\alpha_{\ v} = \frac{1+(-1)^\alpha}2 \frac{c_{\alpha+1}-1}{(2\alpha)!!} \text{ for } \alpha\geq0, \\
&C^v_{\ v} = C^u_{\ u} = C^\alpha_{\ \alpha} = 1 \text{ for } \alpha\leq-1 ,
\eean
all other components being zero. We have
\[
\tilde y_{\hat\alpha}(\zeta) =  \sum_{\hat\gamma\in\hat\Z}  y_{\hat\gamma}(\zeta) \, C^{\hat\gamma}_{\ \hat\alpha}, \quad 
\tilde\theta_{\hat\alpha}(\zeta) = \sum_{\hat\gamma\in\hat\Z} \theta_{\hat\gamma}(\zeta) \, C^{\hat\gamma}_{\ \hat\alpha} \,  \zeta^{\mu_{\hat\gamma}-\mu_{\hat\alpha}} .
\]

The formulas above show that $\tilde y_{\hat\alpha}$ are a Levelt system of deformed flat coordinates. Our choice of $C$  guarantees that they satisfy the orthogonality condition. 
\begin{proposition}
The family of functionals $\left\{ \tilde y_{\hat\alpha}(\hat\la,\zeta) \right\}_{\hat\alpha\in\hat\Z}$ over $M_0 \times C^\times$ forms a Levelt system of deformed flat coordinates on $M_0$ at $\zeta=0$ satisfying the orthogonality condition
\[
\tilde \Theta^*(-\zeta) \tilde\Theta(\zeta) \equiv 1 
\]
where $\tilde\Theta(\zeta)$ is holomorphic in $\zeta$ and $\tilde\Theta(0)=1$. 
\end{proposition}
\begin{proof}
We only need to prove that the orthogonality holds. This is equivalent to showing that 
\beq \label{ort-den}
< d \tilde \theta_{\hat\alpha}(-\zeta) , d \tilde\theta_{\hat\beta}(\zeta) > = \eta_{\hat\alpha \hat\beta}.
\eeq
This is essentially a long computation using the explicit expressions~\eqref{odfc}. We will not reproduce them here. 
\end{proof}
\end{remark}

\subsection{The principal hierarchy}

Recall that the flat metric $\eta$ and the intersection form $\gamma$ on $M_0$ (see~\cite{CDM10}) define two Poisson brackets of hydrodynamic type on the loop space $\cL M_0$ which coincide with those given in Proposition~\ref{prop:poi}.

The Hamiltonian densities $\theta_{\alpha,p}$ define, through the Poisson bracket $\{,\}_1$, an infinite family of commuting flows on $\cL M_0$, which form the Principal hierarchy corresponding to the Frobenius manifold $M_0$. More precisely the Principal hierarchy of $M_0$ is given by the Hamiltonian vector fields on $\cL M_0$
\[
\frac{\pa}{\pa t^{\hat\alpha, p}} \cdot = 
\{ \cdot , H_{\hat\alpha, p} \}_1
\]
where the Hamiltonians are the functionals on $\cL M_0$ defined by
\[
H_{\hat\alpha, p} = \int_{S^1} \theta_{\hat\alpha, p+1} \ dx .
\]

In the first part of this work we have defined a family of commuting vector fields, the extended dispersionless 2D Toda hierarchy, on the loop space $\cL M_1$ of the space $M_1$ of pairs of ``holomorphic'' 2D Toda Lax symbols $(\la(z), \bla(z))$ with a winding numbers condition. Since $M_0$ is an open subset of $M_1$, the extended 2D Toda hierarchy can be restricted to $\cL M_0$. Moreover its Hamiltonian densities and the Poisson  bracket $\{,\}_1$ coincide with those of the Principal hierarchy, hence we clearly have that the two hierarchies coincide. 
\begin{proposition}
The extended dispersionless 2D Toda hierarchy, when restricted to $\cL M_0$, coincides with the Principal hierarchy of the Frobenius manifold $M_0$. 
\end{proposition}

\begin{remark}
Defining the functions on $M_0$
\[
\Omega_{\hat\alpha,p;\hat\beta,q} := \sum_{m=0}^q (-1)^m < \nabla\theta_{\hat\alpha,p+m+1} , \nabla \theta_{\hat\beta,q-m} >
\]
one can easily prove that
\[
\pa_x \Omega_{\hat\alpha,p;\hat\beta,q} = \frac{\pa \theta_{\hat\alpha,p}}{\pa t^{\hat\beta, q}}.
\]
This in particular shows that the Hamiltonian densities $h_{\hat\alpha,p} = \theta_{\hat\alpha,p+1}$ are densities of conserved quantities for all the flows of the hierarchy, and this in turn implies that  the Hamiltonians are in involution w.r.t. both Poisson brackets. 

As usual the symmetry of $\pa_{t^{\hat\alpha, p}} \Omega_{\hat\beta, q; \hat\gamma, r}$ under the exchanges of the three pairs of indices implies that with a solution $\hat\la(t,z)$ of the hierarchy one can associate a tau function such that
\[ \Omega_{\hat\alpha,p;\hat\beta,q} = 
\frac{\pa^2 \log\tau}{\pa t^{\hat\alpha,p} \pa t^{\hat\beta,q}} .
\]
For further details we refer to~\cite{DZ01}.
\end{remark}

\begin{remark}
In Remark~\ref{ort-rem} an alternative choice of deformed flat coordinates $\tilde y_{\hat\alpha}(\zeta)$ was made such that the generating functions of the Hamiltonian densities satisfy the orthogonality condition~\eqref{ort-den}. We call the associated hierarchy 
\[
\frac{\pa}{\pa \tilde t^{\hat\alpha, p}} \cdot = 
\{ \cdot , \tilde H_{\hat\alpha, p} \}_1
\]
with the Hamiltonians 
\[
\tilde H_{\hat\alpha, p} = \int_{S^1} \tilde\theta_{\hat\alpha, p+1} \ dx 
\]
the ``orthogonal'' Principal hierarchy. The Hamiltonian densities $\tilde\theta_{\hat\alpha,p}$ are defined as before by the expansion
\[
\tilde\theta_{\hat\alpha}(\zeta) = \sum_{p\geq0} \tilde\theta_{\hat\alpha,p} \zeta^p . 
\]
Their explicit expression is
\[
\tilde\theta_{\hat\alpha,p} = \frac1{2\pi i} \oint_{|z|=1} \tilde Q_{\hat\alpha,p} \, \frac{dz}z
\]
where
\bes
\label{}
\bea
&\tilde Q_{\alpha,p} = -2  (2\alpha)!! \sum_{0\leq n\leq\frac{p}2} \frac{(\bla-\la)^{p-2n} (\la+\bla)^{2n}}{(2n+\alpha+1)! (p-2n)!}\text{ for }\alpha\geq0, \\
&\tilde Q_{-1,p} = - \frac{(-\la)^p}{p!} \left( \log\left(1+\frac{\bla}{\la} \right) + c_p \right) 
+2^{-p}\sum_{l\geq0}^{p-1} \frac{(\bla-\la)^l (-\bla-\la)^{p-l} c_{p-l}}{l!(p-l)!}, \\
&\tilde Q_{\alpha,p} =Q_{\alpha,p}= - \frac{(\la+\bla)^{\alpha+1}}{\alpha+1} \frac1{p!} \left( \frac{\bla-\la}2 \right)^p \text{ for }\alpha\leq-2,\\
&\tilde Q_{v,p} = - \frac{(-\la)^p}{p!} \left( \log\left( 1+\frac\bla\la \right) +c_p \right) - \sum_{l=0}^{p-1} \frac{(-\la)^l}{l!} \frac{(\la+\bla)^{p-l}}{(2p-2l)!!(p-l)} +\\
&\qquad + \frac{(\bla)^p}{p!} \left( \log( \bla(\la+\bla) ) -c_p \right) + \sum_{l=0}^{p-1} \frac{\bla^l}{l!} \frac{(-\la-\bla)^{p-l}}{(2p-2l)!!(p-l)}, \\
&\tilde Q_{u,p} = Q_{u,p} = \frac{\bla^{p+1}}{(p+1)!}.
\eea
\ees

The ``orthogonal'' Principal hierarchy has a Lax representation
\beq
\label{eq:lax-ort}
\frac{\pa \la}{\ \pa \tilde t^{\hat\alpha,p}} = \{ -(\tilde Q_{\hat\alpha,p})_- , \la \} ,   \qquad
\frac{\pa \bla}{\ \pa \tilde t^{\hat\alpha,p}} = \{ (\tilde Q_{\hat\alpha,p})_+  , \bla \} ,
\eeq
and satisfies the same bi-Hamiltonian recursion relations as before
\bean
&\{ \cdot , \tilde H_{\alpha,p} \}_2 = (\alpha+p+2) \{ \cdot , \tilde H_{\alpha,p+1} \}_1 ,\\
&\{ \cdot , \tilde H_{v,p} \}_2 = (p+1) \{ \cdot , \tilde H_{v,p+1} \}_1 +
2 \{ \cdot , \tilde H_{u,p} \}_1 ,\\
&\{ \cdot, \tilde H_{u,p} \}_2 = (p+2) \{ \cdot , \tilde H_{u,p+1} \}_1 .
\eean

Finally  observe that the ``orthogonal'' densities $\tilde\theta_{\hat\alpha, p}$ are related to the $\theta_{\hat\alpha, p}$ by
\[
\tilde\theta_{\hat\alpha,p} = \sum_{\mu_{\hat\gamma}\leq p +\mu_{\hat\alpha} } \theta_{\hat\gamma,p-\mu_{\hat\gamma}+\mu_{\hat\alpha}} \, C^{\hat\gamma}_{\ \hat\alpha} 
\]
where the matrix $C$ has been defined in Remark~\ref{ort-rem};
note that the sum on the right-hand side is always finite. 

\end{remark}

\section*{Concluding remarks}

In the first part of this article we have defined, by assuming certain analytical properties of the Lax symbols $\la(z)$, $\bla(z)$, a new dispersionless hierarchy which extends the dispersionless 2D Toda hierarchy.
In this direction the most important open problem is the construction of the dispersive extended 2D Toda hierarchy, i.e. a hierarchy containing the difference equations of the 2D Toda hierarchy, introduced in terms of infinite matrices by Ueno and Takasaki~\cite{UT84} or equivalently in terms of difference operators, and an infinite number of extended flows, some of these including logarithmic terms in the spirit of~\cite{CDZ04, C07} and such that its semiclassical limit would coincide with the extended dispersionless 2D Toda defined here.

In the second part of the paper we have considered the relationship of the extended dispersionless 2D Toda hierarchy with the infinite-dimensional Frobenius manifold $M_0$ defined in~\cite{CDM10}. In particular we have constructed the deformed flat connection $\tilde\nabla$ on $M_0 \times \C^\times$ and provided an explicit basis $y_{\hat\alpha}$ of deformed flat coordinates. 
The Principal hierarchy so obtained on $\cL M_0$ coincides with the extended dispersionless 2D Toda.
The analysis of the monodromy at $\zeta=0$ of the $\zeta$-flatness equation indicates that the Frobenius manifold $M_0$ has the typical features of quantum cohomology, including a point of classical limit which explains the simple resonance pattern. An interesting open problem would be to understand if these hints can be extended to a proper (quantum) cohomological interpretation of $M_0$. 

Another important direction of research, will be addressed in subsequent publications, is the study of the properties of the solutions of the principal hierarchy and of their tau functions. Firstly, we plan to study the solution obtained by extending the potential of the Frobenius manifold to the descendent time variables $t^{\hat\alpha,p}$ of the principal hierarchy, the so-called topological solution, which is of particular interest, especially in connection with possible enumerative applications. Secondly, the behavior of a generic solution in the neighborhood of a singular point is expected to have a qualitatively more complicated structure than the $1+1$ case~\cite{D08}, due to the presence of a continuous family of Riemann invariants. 

Related important problems we plan to study are the generalized Stokes phenomenon associated with the behavior of operator-valued linear singular systems on the complex plane and the classification of (classes of) infinite-dimensional Frobenius manifolds. 


{\bf Acknowledgements:} 
G.~C. acknowledges the hospitality of IPhT in Saclay, of IMPA in Rio de Janeiro and of SISSA in Trieste; the support of the ESF-MISGAM exchange grant n.2324, of the INDAM ``Progetto Giovani'' grant and in particular that of Prof.~J.~P.~Zubelli.

L.~Ph.~M. is grateful to Prof.~B.~Dubrovin for being a source of guidance and inspiration during the years of his Ph.D., and for introducing him to the beautiful mathematics of Frobenius Manifolds. He would like to thank J.P. Zubelli for valuable and pleasant discussions, and for giving him the opportunity to join his research group at IMPA. L.~Ph.~M. would also like to acknowledge A. Brini, H. Bursztyn, M. Cafasso, R. Heluani and P. Rossi for insightful discussions. He acknowledges the support of MISGAM for his visit to IPhT, Paris; the support of INDAM for his visit to CMUC, Coimbra and the support of MEC (Minist\'erio da Educa\c c\~ao) and MCT (Minist\'erio da
Ci\^encia e Tecnologia) through CAPES - PNDP (Funda\c c\~ao Coordena\c c\~ao
de Aperfei\c coamento de Pessoal de N\'ivel Superior - Programa
Nacional de P\'os-Doutorado) during his stay at IMPA. Finally, he would like to acknowledge IMPA, for giving him the opportunity of doing mathematics in a professional and friendly environment.

\end{document}